\documentclass[twocolumn]{IEEEtran} 
\usepackage[T1]{fontenc}
\usepackage{color}
\usepackage{mathrsfs}
\usepackage{amsmath}
\usepackage{amsthm} 
\usepackage{amssymb}
\usepackage{graphicx}
\usepackage{cite}
\usepackage{makecell}
\usepackage{algorithm}
\usepackage{algorithmic}
\usepackage{fancyhdr}
\usepackage{diagbox}
\usepackage{bm}
\usepackage{lettrine}
\usepackage[unicode=true,
bookmarks=true,bookmarksnumbered=true,bookmarksopen=true,bookmarksopenlevel=1,
breaklinks=false,pdfborder={0 0 0},pdfborderstyle={},backref=false,colorlinks=false]
{hyperref}
\hypersetup{
pdftitle={Your Title},
pdfauthor={Your Name},
pdfpagelayout=OneColumn, pdfnewwindow=true, pdfstartview=XYZ, plainpages=false}

\usepackage{multirow} 
\usepackage{xcolor}

\allowdisplaybreaks[4]

\usepackage[caption=false,font=footnotesize]{subfig}

\IEEEoverridecommandlockouts

\usepackage{geometry}
\usepackage{cleveref}

\newtheorem{lemma}{Lemma}

\begin{document}

\title{\textcolor{black}{Distributionally Robust Optimization for Aerial Multi-access Edge Computing via Cooperation of UAVs and HAPs}}
\author{\IEEEauthorblockN{Ziye Jia, \IEEEmembership{Member, IEEE,} Can Cui, Chao Dong, \IEEEmembership{Member, IEEE, }Qihui Wu, \IEEEmembership{Fellow, IEEE, }\\
Zhuang Ling, \IEEEmembership{Member, IEEE, }Dusit Niyato, \IEEEmembership{Fellow, IEEE, }and Zhu Han, \IEEEmembership{Fellow, IEEE}}
\thanks{This work was supported  in part by Jiangsu Province Frontier Leading Technology Basic Research Project under Grant BK 20222013,  in part by National Natural Science Foundation of China under Grant 62301251, in part by the Natural Science Foundation of Jiangsu Province of China under Project BK20220883, in part by  the open research fund of National Mobile Communications Research Laboratory, Southeast University (No. 2024D04), in part by the Aeronautical Science Foundation of China 2023Z071052007, in part by the Young Elite Scientists Sponsorship Program by CAST 2023QNRC001, and partially supported by NSF  ECCS-2302469, CMMI-2222810, Toyota. Amazon and Japan Science and Technology Agency (JST) Adopting Sustainable Partnerships for Innovative Research Ecosystem (ASPIRE) JPMJAP2326. (\textit{Corresponding authors: Can Cui and Chao Dong.})}
\thanks{Ziye Jia is with the College of Electronic and Information Engineering, Nanjing University of Aeronautics and Astronautics, Nanjing 211106, China, and also with the National Mobile Communications Research Laboratory, Southeast University, Nanjing, Jiangsu, 211111, China (e-mail: jiaziye@nuaa.edu.cn).}
\thanks{Can Cui, Chao Dong and Qihui Wu are with the College of Electronic and Information Engineering, Nanjing University of Aeronautics and Astronautics, Nanjing 211106, China (e-mail: cuican020619@nuaa.edu.cn; dch@nuaa.edu.cn; wuqihui@nuaa.edu.cn).}
\thanks{Zhuang Ling is with the College of Communication Engineering, Jilin University, Changchun 130012, China (e-mail: lingzhuang@jlu.edu.cn).}
\thanks{Dusit Niyato is with the School of Computer Science and Engineering, Nanyang Technological University, Singapore 639798 (e-mail: dniyato@ntu.edu.sg).}
\thanks{Zhu Han is with the University of Houston, Houston, TX 77004 USA, and also with the Department of Computer Science and Engineering, Kyung Hee University, Seoul 446-701, South Korea (e-mail: hanzhu22@gmail.com).}
}
\maketitle

\begin{abstract}
With an extensive increment of computation demands, the aerial multi-access edge computing (MEC), mainly based on unmanned aerial vehicles (UAVs) and high altitude platforms (HAPs), plays significant roles in future network scenarios. In detail, UAVs can be flexibly deployed, while HAPs are characterized with large capacity and stability. Hence, in this paper, we provide a hierarchical model composed of an HAP and multi-UAVs, to provide aerial MEC services. Moreover, considering the errors of channel state information from unpredictable environmental conditions, we formulate the problem to minimize the total energy cost with the chance constraint, which is a mixed-integer nonlinear problem with uncertain parameters and intractable to solve. To tackle this issue, we optimize the UAV deployment via the weighted K-means algorithm. Then, the chance constraint is reformulated via the distributionally robust optimization (DRO). Furthermore, based on the conditional value-at-risk mechanism, we transform the DRO problem into a mixed-integer second order cone programming, which is further decomposed into two subproblems via the primal decomposition. Moreover, to alleviate the complexity of the binary subproblem, we design a binary whale optimization algorithm. Finally, we conduct extensive simulations to verify the effectiveness and robustness of the proposed schemes by comparing with baseline mechanisms.
\end{abstract}

\begin{IEEEkeywords}
	Aerial multi-access edge computing, resource allocation, distributionally robust optimization (DRO), conditional value-at-risk (CVaR), primal decomposition, binary whale optimization (BWOA).
\end{IEEEkeywords}
\section{Introduction}
\lettrine[lines=2]{I}{N} recent years, with the extensive growth of computational intensive tasks, the multi-access edge computing (MEC) technique is raised to provide services for various ground users (GUs) in the sixth generation (6G) communication networks \cite{Intelligent-Guo,Servic-Tian,Blockchain-Xu,Joint-Zhao}. However, as for the GUs in remote areas such as deserts and oceans, it lacks ground infrastructures to provide communication coverage and MEC services \cite{LEO-Jia,AoI-Zhang}. As forecasted that the global unmanned aerial vehicles (UAVs) market is expected to reach the worth of \$55.8 billion by 2030, UAVs show the colossal potential in the applications of future industries \cite{Toward-Bai}. The non-terrestrial networks can provide ubiquitous coverage for the remote GUs, in which UAVs can be flexibly and quickly deployed on demand with low cost \cite{Unmanned-Zhan, Adaptive-Chen,Multi-UAV-Guo, Computation-You, Deep-Wang, Adaptive-Chen-Wen}. However, UAVs are limited by the load capacity of computing module, battery, etc. Alternatively, the high altitude platform (HAP), suspending above 20km at a quasi-static position, and equipped with stronger computing resources and sufficient energy, can compensate for the resource-limited UAVs \cite{Communication-Karabulut}, \cite{Vision-Karabulut}. Besides, compared with HAPs, UAVs are relatively flexible and can be rapidly deployed to meet the sudden surge in data requests. Therefore, in the context of 6G networks, the cooperation of UAVs and HAPs can provide flexible and stable aerial MEC for GUs in diverse applications to reduce latency and improve the quality of service (QoS).

Unfortunately, in the aerial MEC network, the tasks generated from GUs may be heterogeneous with different QoS demands such as the tolerated delay. Besides, the resources of aerial MEC networks, such as communication, computing, and energy are limited, in which the energy supply is the basic for all operations \cite{Resource-Li}. Hence, how to guarantee the QoS of GUs and take full advantage of aerial resources is a key issue for ubiquitous communication and computation services in the 6G networks\cite{Mobile-Zhou}, \cite{Intelligent-Dong}. Furthermore, considering the unpredictable environmental fluctuations, the communication link from the GU to UAV (G2U) is highly dynamic, and the channel state information (CSI) is imperfect, which may cause errors and mismatches between the realistic situation and ideal circumstance \cite{Energy-Sheng}. Such errors bring more challenges for the resource allocation scheme in the aerial MEC. Besides, how to cooperate UAVs and HAPs for efficient data offloading and resource optimization is also challenging.

To deal with the above challenges, in this paper, we propose an aerial MEC framework composed of UAVs and an HAP to cooperatively serve the GUs in remote areas, with the consideration of imperfect transmission CSI. In detail, an uncertainty set is constructed to capture the potential random parameters and a chance constraint for task latency with CSI estimation errors is formulated. Then, considering multi-resource constraints for UAVs and the HAP, we formulate the problem to minimize the total energy consumption, with regard to UAV positions, GU-UAV connection decisions, offloading strategies and resource allocation. Since the problem is in the form of mixed integer non-linear programming (MINLP) and NP-hard to solve \cite{Joint-Yu}, we firstly cluster the GUs and deploy UAVs at appropriate positions via proposing the weighted K-means deployment (WKD) based algorithm with a low time complexity. Taking into account the different characteristics of tasks, the weighted distance metric is applied so that the importance of different tasks is incorporated. Then, we reformulate the chance constraint without distribution information into a mixed integer second order cone programming (MISOCP) form by employing the distributionally robust optimization (DRO) and conditional value-at-risk (CVaR) mechanism. To reduce the complexity, via the primal decomposition, we further decompose the MISOCP problem into two subproblems with respect to the offloading decisions and computing resource allocation, respectively. The problem concerning resource allocation is solved by a standard convex toolkit. Moreover, to tackle the integer programming problem related to the offloading decision, we design a meta-heuristic algorithm termed as binary whale optimization algorithm (BWOA).

The main contributions of this work are summarized as follows.
\begin{itemize}
	\item We propose a hierarchical aerial MEC model composed of an HAP and multi-UAVs to provide services for remote GUs, in which UAVs can be deployed flexibly and the HAP provides stable and strong computing services. Besides, the CSI estimation error is modeled by an uncertainty set based on the historical statistical information and the time latency requirement is formulated as a chance constraint.
	
	\item To handle the problem of multi-UAV deployment, a WKD based algorithm is designed. Then, by the DRO and CVaR based mechanism, the chance constraint is reformulated into an MISOCP form.
	
	\item To tackle the reformulated mixed integer programming (MIP) problem with the MISOCP constraint, by leveraging the primal decomposition, it is decomposed into two subproblems. The subproblem on the resource allocation is convex and solved via CVX. To further reduce the complexity of the binary offloading subproblem, we design the BWOA.
	
	\item Extensive simulations are conducted to evaluate the proposed algorithms under various circumstances. The robustness of the designed algorithms with CSI estimation errors is verified. Moreover, by comparing with other baseline algorithms, the effectiveness and low-complexity of the proposed algorithms are verified.
\end{itemize}

The rest of this paper is arranged as follows. Related works are presented in Section \ref{sec2}. Section \ref{sec3} proposes the system model and problem formulation. Algorithms are designed in Section \ref{sec4}. Simulations and numerical results are provided in Section \ref{sec5}. Finally, we draw conclusions in Section \ref{sec6}.
\section{Related Works}\label{sec2}
As for the UAV-based MEC, there exist abundant recent researches. For instance, \cite{Multi-Agent_Zhao} proposed a collaborative MEC frame and exploited a deep reinforcement learning method to jointly optimize resource allocation, UAV trajectory, and task scheduling. \cite{Multi-UAV-Zhan} discussed a multi-UAV-enabled MEC system and solved the decoupled two subproblems via alternating optimization and successive convex approximation mechanisms. Considering the competitive relationship among UAVs, \cite{Joint_Li} formulated a joint optimization problem for the multi-dimensional resource constrained UAV-MEC network and put forward a triple learner based approach. While these studies made significant progresses in the resource and trajectory optimization, improved frameworks are still worth considering. In \cite{Resource_Liu}, the authors designed a three-stage alternating algorithm to address issues concerning energy consumption in the UAV-related MEC system. In \cite{Multi-Agent-Ning}, a deep reinforcement learning approach was devised to minimize the computation cost in the multi-UAV based MEC system. \cite{Evolutionary-Song} investigated a multi-objective optimization problem in the MEC network to minimize the delay and energy consumption as well as maximize the number of collected tasks of UAVs. In \cite{Computation_Liu}, a hierarchical UAV-assisted MEC framework was studied to minimize the sum of latency and energy consumption, in which a method by jointly combining deep reinforcement learning and convex optimization was designed. \cite{Optimization_Luo} presented a two-layer optimization framework to reduce the energy consumption for the UAV-based MEC. The authors in \cite{Joint_Wang} designed a two-layer optimization approach and tackled the 0-1 integer programming problem by a greedy algorithm, in which by jointly optimizing task scheduling and UAV locations, the energy consumption in the MEC system was minimized. \cite{Robust_Li} designed a robust multi-agent approximation strategy to address the uncertainties of CSI and task complexity, which was solved by the multi-agent deep reinforcement learning. However, limited by their battery capacities, the applications of UAVs are restricted in terms of the large-scale service provision for delay sensitivity and computation intensive tasks.

Although the above researches have made great progresses in the resource and trajectory optimization, UAVs are limited by their battery capacities and have shortcomings in providing better services for delay sensitive and computationally intensive tasks. Different from UAVs with constrained capabilities, HAPs can provide strong payloads and stable coverage, which contributes to completing intensive MEC services. In recent years, some works begin to explore the applications of HAPs based MEC services. For example, \cite{Design_Granelli} jointly deployed multi-UAVs and an HAP to provide connectivity and computing service for GUs. \cite{Hierarchical_Jia} focused on data offloading in the UAV-HAP MEC system and developed a matching-based algorithm to maximize the total processed data. In \cite{AI-based_Maryam}, a multi-dimensional resource allocation problem in the UAV-HAP system was designed to minimize the average age of information in response to the uncertain errors of CSI, and a learning-based algorithm was presented to tackle this non-convex problem. Due to the limited resources of the aerial MEC platforms, more studies focused on the issue to improve the energy consumption and resource utilization. In \cite{Energy-Chen}, a resource allocation problem minimizing the energy cost was studied for the UAV-HAP assisted MEC, and solved by the distributed online algorithm based on the game theory. The authors in \cite{Energy-Song} focused on the energy-efficient trajectory optimization problem in the UAV-HAP based MEC system and designed a modified multi-objective reinforcement learning algorithm. \cite{Cooperative_Cao} employed the K-means and multi-agent reinforcement learning algorithms for resource utilization in the UAV-HAP assisted MEC system. The authors in \cite{Deep_Cheng} investigated a computation offloading problem in the UAV-HAP aerial MEC system, and a markov game was conducted to enhance the energy harvesting performance for UAVs. \cite{AoI-Song} built a multi-objective Markov decision process model towards the age of information and energy tradeoff problem in which UAVs and HAPs cooperatively provided MEC services for ground devices.

As analyzed above, the cooperation of UAVs and HAPs leverages the advantages of the aerial platforms for better energy efficiency, lower latency, and higher capability. Nevertheless, in these studies, the CSI errors are mostly ignored or following a specific distribution, which is impractical due to the various interferences caused by actual environment. The inability to estimate accurate CSI in the practical scenarios brings more challenges for efficient resource allocations. Therefore, it is essential to exploit robust algorithms to deal with the unpredictable fluctuation from the environment. Based on above considerations, in this paper, we focus on the cooperation of UAVs and the HAP to provide robust aerial MEC services for GUs, with the consideration of the imperfect CSI by the uncertainty set.

\section{System Model and Problem Formulation}\label{sec3}

In this section, a two-layer aerial MEC model is proposed in Section \ref{subsection-1}. Then, the communication model and offloading model are proposed in Sections \ref{subsection-2} and \ref{subsection-3}, respectively. The problem formulation is detailed in Section \ref{subsection-4}.

\subsection{Aerial MEC Model}\label{subsection-1}

As shown in Fig. \ref{model}, a hierarchical aerial MEC system is proposed. $M$ GUs indicated by the set $\mathcal{M}=\{1,2,\ldots,M\}$, $m\in\mathcal{M}$, are randomly distributed within the remote areas. $N$ UAVs equipped with edge servers, the set of which is denoted by $\mathcal{N}=\{1,2,\ldots,N\}$, $n\in\mathcal{N}$, are deployed to provide MEC services. An HAP, denoted by $h$, hovers at a fixed position and acts as the supplement for the resource-limited UAVs. The Cartesian coordinate is utilized to represent the locations. $\mathbf{w}_m$ denotes the location of GU $m$, where $\mathbf{w}_m=(x_m,y_m)$. All UAVs are assumed to hover at the same altitude $z_n$, and the horizontal deployment position of UAV $n$ is denoted by $v_n=(x_n,y_n)$. After deployment, the UAVs hover in the air and are regarded as quasi-stationary. The distance between GU $m$ and UAV $n$ is calculated as:
\begin{equation}
	d_{m,n}\hspace{-0.25em}=\hspace{-0.25em}\sqrt{(x_m-x_n)^2+(y_m-y_n)^2+z_n^2},\forall m \in \mathcal{M}, \forall n\in\mathcal{N}.\label{equation1}
\end{equation}
Since UAVs are equipped with limited computing resources, an HAP with strong capabilities is deployed in the upper layer to assist with computing, whose location is denoted by $\varpi _h=(x_h,y_h,z_h)$. Therefore, the distance between UAV $n$ and HAP $h$ is
\begin{equation}
	d_{n,h}=\sqrt{(x_n-x_h)^2+(y_n-y_h)^2+(z_n-z_h)^2}, \forall n\in\mathcal{N}.
\end{equation}
\begin{figure}
	\centering{\includegraphics[width=0.85\linewidth]{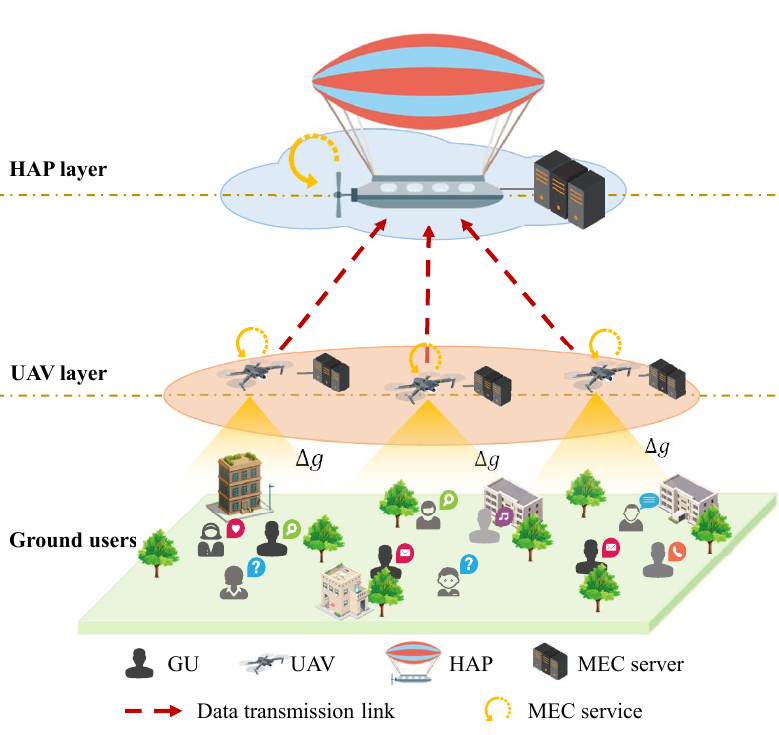}}
	\caption{Aerial MEC model composed of UAVs and an HAP.}
	\label{model}
\end{figure}

The task generated from GU $m$ is denoted as $\left(L_m,c_m,T_m^{max}\right)$, where $L_m$ represents the size of task data, $c_m$ indicates the required CPU cycles to process 1bit data and $T_m^{max}$ is the maximum tolerable delay of task $m$. To guarantee the delay limitation, the task needs to be completed within $T_m^{max}$. We consider that the task cannot be divided and the task offloading pattern is the binary mode. Let binary variable $\delta_m^n$ indicate the connection relationship between GU $m$ and UAV $n$, i.e.,
\begin{equation}
	\delta_m^n=
	\begin{cases}
		1, \quad \text{GU }m\text{ is connected with UAV }n,\\
		0, \quad \text{otherwise.}
	\end{cases}
\end{equation}
Considering the accessing constraint, each GU can only connect to one UAV, we have
\begin{equation}\label{c7}
	\sum_{n=1}^N\delta_m^n=1, \forall m\in \mathcal{M}.
\end{equation}

If a UAV is not able to provide sufficient computing resources for the task, or the delay limitation is unable to be satisfied, then the task is forwarded to the HAP for processing. In this case, the UAV performs as a relay. In detail, binary variable $\lambda_m^n$ is introduced to indicate whether the task collected by UAV $n$ is forwarded to the HAP, i.e.,
\begin{equation}
	\lambda_m^n=\begin{cases}
		1, \quad\text{task }m\text{ is forwarded to the HAP by UAV }n,\\
		0, \quad\text{otherwise.}
	\end{cases}
\end{equation}
The HAP is assumed to process at most $H$ tasks simultaneously, and so we have
\begin{equation}\label{c8}
	\sum_{m=1}^M\sum_{n=1}^N\lambda_m^n\leq H.
\end{equation}

\subsection{Communication Model}\label{subsection-2}
\subsubsection{G2U Channel Model} The G2U channel is a large-scale fading model \cite{Service-Zheng} and can be regarded as a line-of-sight channel. The uplink transmission channel gain under ideal condition of GU $m$ is given as:
\begin{equation}
	\bar{g}_m^u=\sum_{n=1}^N\frac{\delta_m^ng_0^u}{d_{m,n}^2},\forall m\in \mathcal{M},
\end{equation}
where $g_0^u$ is the power gain at the reference distance $d_0=1$m.

Since the G2U channel is time-varying and vulnerable with the impacts from obstacles, complicated terrains and electromagnetic interferences, the CSI cannot be obtained precisely. In other words, there exist CSI estimation errors between the ideal and realistic environments, due to the inevitable disturbances and interferences from the environment. Accordingly, we denote the actual G2U channel gain as $g_{m}^u$:
\begin{equation}
	g_{m}^u=\bar{g}_{m}^u+\Delta_{m},\forall m\in \mathcal{M},
\end{equation}
where $\Delta_{m}$ is the unmeasurable CSI estimation error. Clearly, it is difficult to obtain accurate results or probability distributions of the CSI errors in real situations. Therefore, we consider that the moment estimation information can be obtained from the historical statistical data. In particular, the uncertainty set $\mathcal{P}$ is constructed to describe all the possible distributions of the random errors, i.e.,
\begin{equation}
	\mathcal{P}=\biggl\{\mathbb{P}\in\mathcal{P}\Bigg|
		\begin{aligned}
			\mathbb{E_P}(\Delta_{m})=\mu_m,\\
	        \mathbb{D_P}(\Delta_{m})=\sigma_m^{2},
		\end{aligned}
	\biggr\},
\end{equation}
where $\mu_m$ is the mean of random parameters $\Delta_{m}$ under distribution $\mathbb{P}$, and $\sigma_m^{2}$ is the corresponding variance. The uncertainty set $\mathcal{P}$ comprises all possible probability distributions of the random CSI estimation error $\Delta_m$, i.e., $\mathbb{P}\in\mathcal{P}$. Besides, to simplify the communication model, we adopt the orthogonal frequency division multiple access (OFDMA) technology. In this way, GUs are enabled to transmit their data simultaneously, and the mutual interference is correspondingly ignored. According to the Shannon formula, the uplink rate of G2U channel is
\begin{equation}
	r_{m}^u=B_u\log_2\left(1+\frac{p_ug_{m}^u}{n_0B_u}\right),\forall m\in \mathcal{M},
\end{equation}
where $B_u$ denotes the bandwidth allocated to each task, $p_u$ is the transmitting power of GUs, and $n_0$ represents the power spectrum density of additive white noise. Then, the uplink transmission delay of GU $m$ is
\begin{equation}
	t^u_{m}=\frac{L_m}{r^u_m},\forall m\in \mathcal{M}.
\end{equation}
\subsubsection{U2H Channel Model}
Different from the vulnerable G2U link, there are few obstacles or environment reflection disturbances in the UAV-to-HAP (U2H) link. Therefore, characterized with a wider view, we consider that the U2H link is estimated precisely \cite{AI-based_Maryam}. Moreover, the OFDMA is adopted in the U2H channel to avoid interferences. Consequently, considering the free space loss and rain attenuation, the maximum achievable rate from UAV $n$ to HAP $h$ is \cite{Joint-Li,Two-Li,Cooperative-Kang}
\begin{equation}\label{(10)}
	r_{n}^h=B_h\log_2\left(1+\frac{p_h g_{n}^hL_sL_l}{k_BT_0B_h}\right),\forall n\in \mathcal{N},
\end{equation}
where $p_h$ and $g_{n}^h$ denote the transmission power and antenna power gain between UAV $n$ and HAP $h$, respectively. $d_{n,h}$ is the distance between UAV $n$ and HAP $h$. Moreover, $L_s=\left(\frac{v_c}{4\pi d_{n,h}f_c}\right)^2$ is the free space path loss. $L_l$ is the total line loss. $k_B$ is the Boltzmann's constant. $T_0$ is the system noise temperature. $B_h$ denotes the bandwidth. $f_c$ represents the center frequency. $v_c$ is the speed of light. Therefore, the transmission latency and energy consumption for task $m$ from UAV $n$ to HAP $h$ are calculated as:
\begin{equation}\label{(11)}
	t^h_{m,n}=\frac{\lambda_m^nL_m}{r_{n}^h},\forall m\in \mathcal{M},\forall n\in \mathcal{N},
\end{equation}
and
\begin{equation}
	E^h_{m,n}=p_ht_{m,n}^h,\forall m\in \mathcal{M},\forall n\in \mathcal{N},
\end{equation}
respectively. Since the backhaul data is much smaller than uplink data, the backhaul delay is ignored \cite{Cooperative-Liu}.

\subsection{Computation  Model}\label{subsection-3} 
\subsubsection{UAV-based Computation Model}
	Let $f_{m}$ represent the CPU frequency allocated to task $m$. Recall that $L_m$ denotes the data size and $c_m$ is the required number of CPU cycles to compute 1bit data. As a result, the computation latency for processing task $m$ is
	\begin{equation}
		t_m^{cu}=\frac{c_mL_m}{f_{m}},\forall m\in \mathcal{M}.\label{temp1}
	\end{equation}
	Based on \cite{UAVAssisted-Liu}, the energy consumption for handling task $m$ is
	\begin{equation}
		E_m^{cu}=\sum_{n=1}^N(\delta_m^n-\lambda_m^n)\varepsilon_n c_mL_mf_{m}^{2},\forall m\in \mathcal{M},
	\end{equation}
	where $\varepsilon_n$ is the effective switched capacitance related to the architecture of MEC servers on UAVs. Note that the CPU frequency of the MEC server is constrained:
	\begin{equation}\label{c5}
		\sum_{m=1}^M (\delta_m^n-\lambda_m^n)f_{m}\leq F_{max}^n,\forall n\in \mathcal{N},
	\end{equation}
	where $F_{max}^n$ is denoted as the maximum CPU cycle frequency of the UAV $n$.

\subsubsection{HAP-based Computing Model}
	Recall that binary variable $\lambda_m^n$ represents whether task from GU $m$ is computed at the HAP. Then, in the HAP based computation model, the computing delay and energy consumption for handling task $m$ are
	\begin{equation}
		t_m^{ch}=\frac{c_mL_m}{f_{m}},\forall m\in \mathcal{M},
	\end{equation}
	and
	\begin{equation}
		E_m^{ch}=\sum_{n=1}^N\lambda_m^n\varepsilon_h f_{m}^{2}c_mL_m,\forall m\in \mathcal{M},
	\end{equation}
	respectively, in which $\varepsilon_h$ is the energy consumption coefficient related to the specific chip structure of an MEC server \cite{PDDQNLP-Lin}. Let $F_{max}^h$ denote the maximum computational rate of HAP, the CPU frequency constraint for MEC server of the HAP is
	\begin{equation}\label{c6}
		\sum_{m=1}^M \sum_{n=1}^N \lambda_m^n f_{m}\leq F_{max}^h.
	\end{equation}

Based on the above discussion, the total delay for computing task $m$ is related to the transmission and computation, i.e.,
\begin{equation}
	\begin{split}
	t_m^{total}=&t_{m}^u+\sum_{n=1}^N(\delta_m^n-\lambda_m^n)t_m^{cu}+\sum_{n=1}^N\lambda_m^nt_{m,n}^h\\
	&+\sum_{n=1}^N\lambda_m^nt_{m}^{ch},m\in\mathcal{M}.
	\end{split}
\end{equation}
Moreover, since UAVs are hovering in the air after deployment, the energy consumption for hovering is constant. Therefore, the remaining energy consumption for UAV $n$ is for transmission and computation \cite{Computation_Liu}, i.e.,
\begin{equation}
	E_{n}^{total}=\sum_{m=1}^M E_{m,n}^h+\sum_{m=1}^M E_m^{cu},\forall n\in\mathcal{N}.
\end{equation}
Besides, since the HAP hovers at the quasi-position, the remaining energy consumption of the HAP is for computation:
\begin{equation}
	E_{h}^{total}=\sum_{m=1}^ME_m^{ch}.
\end{equation}
\subsection{Problem Formulation}\label{subsection-4}

To deal with the potential uncertainties without distribution information, we formulate $\textbf{P0}$ with the chance constraints to minimize the total energy cost of the aerial MEC platforms, with restrictions of UAV deployment, task offloading, and resource limitation, i.e.,
\begin{subequations}
	\begin{align}
		\textbf{P0:} \quad &\min_{\mathbf{v},\bm{\delta},\bm{\lambda},\mathbf{f}} \sum_{n=1}^N E_n^{total}+E_h^{total}\nonumber\\
		\textrm{s.t.}\quad&\mathbf{Pr}\left\{t_m^{total}\leq T_m^{max}\right\}\geq\alpha_m,\forall m\in \mathcal{M},\label{c1}\\
		&\lambda_m^n\leq\delta_m^n,\forall m\in \mathcal{M}, n\in \mathcal{N},\label{c2}\\
		&E_{n}^{total}\leq E_n^{max},\forall n\in \mathcal{N},\label{c3}\\
		&E_{h}^{total}\leq E_h^{max},\label{c4}\\
		&v_n\in \biggl \{(x_n, y_n) \Bigg|\begin{aligned}X^{\min} \leq x_n \leq X^{\max}\\Y^{\min} \leq y_n \leq Y^{\max}\end{aligned} \biggl \},\forall n\in\mathcal{N},\label{c12}\\
		&\delta_m^n\in\{0,1\},\forall m\in \mathcal{M}, n\in \mathcal{N},\label{c9}\\
		&\lambda_m^n\in\{0,1\},\forall m\in \mathcal{M}, n\in \mathcal{N},\label{c10}\\
		&f_m\geq 0,\forall m\in \mathcal{M},\label{c11}\\
		&(\ref{c7}),(\ref{c8}),(\ref{c5}),(\ref{c6}),\nonumber
	\end{align}
\end{subequations}
wherein $\mathbf{v}=\{v_n|\forall n\}$ is the UAV deployment positions, $\bm{\delta}=\{\delta_m^n|\forall m,\forall n\}$ represents GU-UAV connection relationships, $\bm{\lambda}=\{\lambda_m^n|\forall m,\forall n\}$ denotes the task offloading indicators and $\mathbf{f}=\{f_{m}|\forall m\}$ is the resource allocation schemes. With respect to the uncertain CSI estimation errors, (\ref{c1}) is the chance constraint under the uncertainty set $\mathcal{P}$, which indicates that the total latency for processing task $m$ should not be larger than $T_m^{max}$ with a probability of $\alpha_m$. Constraint (\ref{c2}) is the data flow conservation, reflecting the inherent relationship between $\lambda_m^n$ and $\delta_m^n$. Constraints (\ref{c3}) and (\ref{c4}) indicate the total energy consumption of UAV and HAP should not be larger than the maximum capacities $E_n^{max}$ and $E_h^{max}$, respectively. (\ref{c12}) is the constraint for the deployment range of UAVs, wherein $[X^{min},X^{max}]$ and $[Y^{min},Y^{max}]$ are the horizontal and vertical bounds of the area, respectively.

It is observed that $\textbf{P0}$ is related with the random parameter $\Delta_m$ under uncertainty set $\mathcal{P}$ without distribution information. Besides, $\textbf{P0}$ is an MINLP concerning binary variables $\bm{\delta}$ and $\bm{\lambda}$, and continuous variables $\mathbf{f}$ and $\mathbf{v}$, and the time complexity is exponential with the problem scale growing. Therefore, solving $\textbf{P0}$ with efficiency is  intractable. 

\section{Algorithm Design}\label{sec4}

To tackle $\textbf{P0}$ efficiently, we divide the process into two phases of UAV deployment and computation offloading. For clarity, the overview of the designed algorithms is illustrated in Fig. \ref{algorithm}. As for the UAV deployment, we design a WKD based algorithm in Section \ref{subsection-5} to obtain UAV positions $\mathbf{v}$ and GU-UAV connections $\bm{\delta}$ with a low time complexity. Then, based on the determined pre-deployment of UAVs to handle the CSI estimation error, the CVaR-based mechanism is proposed in Section \ref{subsection-6}. Thus, problem $\textbf{P1}$ with the chance constraint (\ref{c1}) is conservatively approximated and reformulated into $\textbf{P2}$ via the DRO and CVaR based mechanism. In Section \ref{subsection-7}, the reformulated problem $\textbf{P2}$ is dealt with by the primal decomposition. $\textbf{P3}$ is in the form of SOCP and can be solved via CVX. Furthermore, to effectively obtain the integer offloading strategies of the subproblem $\textbf{P4}$, it is reformulated into $\textbf{P5}$, and the BWOA is designed in Section \ref{subsection-8}.

\begin{figure}
	\centering{\includegraphics[width=0.95\linewidth]{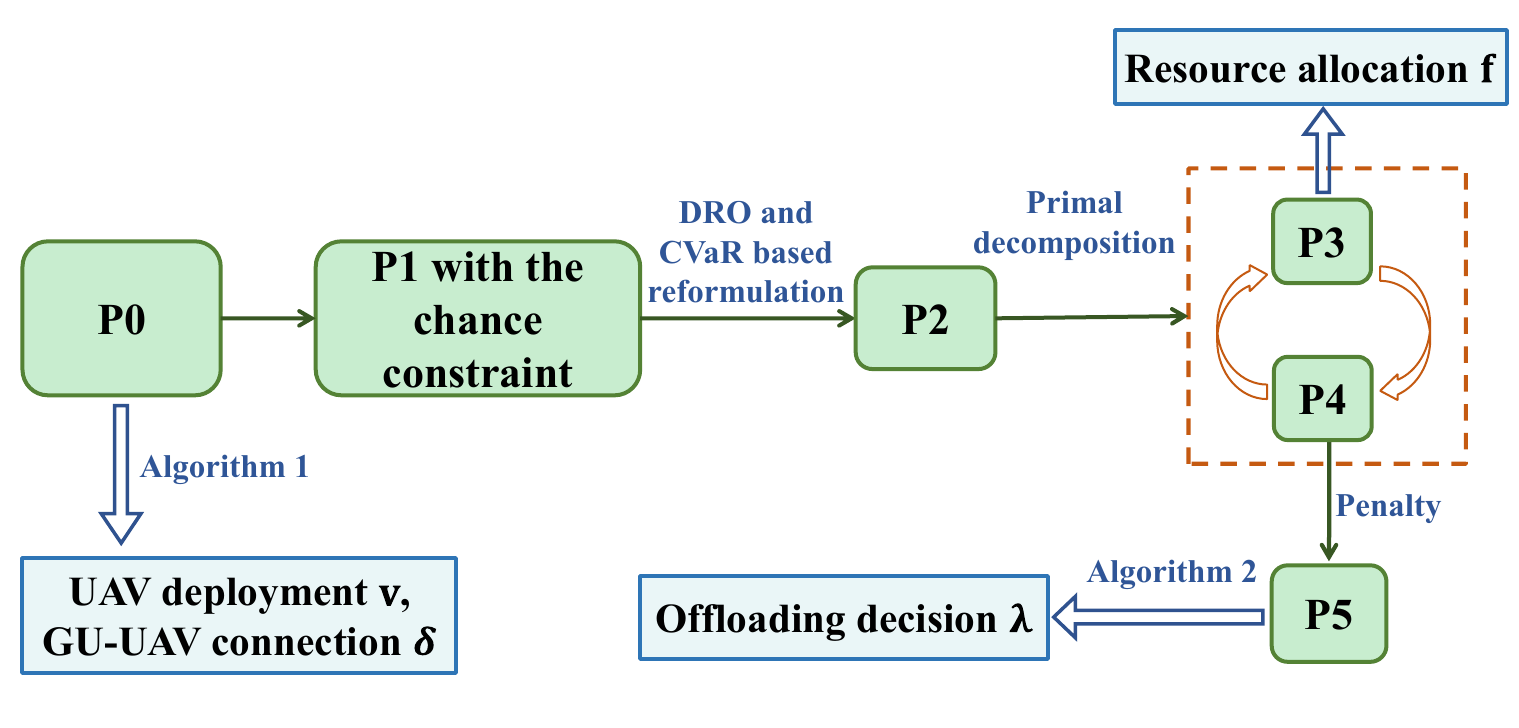}}
	\caption{Overview of the designed algorithms.}
	\label{algorithm}
\end{figure}

\subsection{UAV Deployment Optimization}\label{subsection-5}

Generally, to reduce the G2U transmission delay, UAVs should be deployed closer to GUs. Moreover, considering the various tasks with different inherent characteristics including data size $L_m$, computation complexity $c_m$, and the maximum tolerable delay $T_m^{max}$, if the UAV is deployed closer to GUs with larger load, the latency is further reduced to better satisfy the QoS. Hence, we design the WKD mechanism to obtain the deployment position $\mathbf{v}$ of UAVs and the GU-UAV connections $\bm{\delta}$, which highlights the importance of time-sensitive tasks and provides a more practical and efficient solution for the pre-deployment of UAVs. Since the potential uncertainties have a relatively small impact, they are ignored during the pre-deployment operations. The detailed WKD algorithm is provided in Algorithm \ref{alg1:K-means}.

Firstly, we select $N$ points in the area as initial positions for UAVs. Then, the distance between UAVs and GUs is obtained according to (\ref{equation1}) and all GUs are accordingly assigned to their nearest clusters (line \ref{line-4}). Then, the center point of each cluster is recalculated and the UAV positions are updated as (line \ref{line-6}):
\begin{equation}\label{equation47}
	v_n = \frac{\sum\limits_{m\in\mathcal{U}_n} \iota_m \mathbf{w}_m}{\sum\limits_{m\in\mathcal{U}_n}\iota_m},\forall n \in \mathcal{N},
\end{equation}
where the weight coefficient $\iota_m$ of the task from GU $m$ is obtained via $\iota_m = \varsigma_1 L_m + \varsigma_2 c_m + \frac{1-\varsigma_1 -\varsigma_2}{T_m^{max}}$. Specifically, $\varsigma_1$ and $\varsigma_2$ are weighted variables for a tradeoff among $L_m$, $c_m$ and $T_m^{max}$. $\mathcal{U}_n$ denotes the set of GUs belonging to cluster $n$. Then, this process is repeated until the result converges and $I_1^{max}$ is denoted as the number of iterations. After $\mathbf{v}$ is obtained, GUs are accordingly connected with their corresponding UAVs (line \ref{line-12}). During each iteration, the distances between GUs and UAVs are calculated and the positions of UAVs are updated towards convergence \cite{K-Means-Selim}. Moreover, since the time complexity of Algorithm \ref{alg1:K-means} is related to the scale of GU $M$ and UAV $N$, the corresponding time complexity is $\mathcal{O} (MNI^{max}_1)$. During the clustering process, each user is assigned to a cluster and the UAVs are deployed at the weighted cluster centers. Leveraging the proposed WKD algorithm, which is operated based on the distribution of GUs, we obtain the pre-deployment for $N$ UAVs to cover all the clusters as the initial positions for the subsequent operations.

\subsection{CVaR-based Mechanism for Chance Constraint}\label{subsection-6}

\begin{algorithm}[t]
	\caption{Weighted K-means Based Multi-UAV Deployment}\label{alg1:K-means}
	\begin{algorithmic}[1]
	\REQUIRE Locations of GUs $\mathbf{w}_m$.
	\STATE \textit{Initialization:} Set initial $\mathbf{v}$, $\delta_m^n=0,\forall m\in\mathcal{M},\forall n\in\mathcal{N}$.
	\REPEAT
		\FOR{$n\in\mathcal{N}$}
			\STATE {Calculate the distance $d_{m,n}$ between GU $m$ and UAV $n$ based on (\ref{equation1}).\label{line-4}}
			\STATE {Assign GU $m$ to its nearest UAV $n^*$.\label{line-5}}
			\STATE {Update $v_n$ based on (\ref{equation47}).\label{line-6}}
		\ENDFOR
	\UNTIL{the result converges.}
	\FOR{$n\in\mathcal{N}$}
		\FOR{$m\in\mathcal{U}_n$}
			\STATE {Connect GU $m$ with UAV $n$, i.e., $\delta_m^n=1$.\label{line-12}}
		\ENDFOR
	\ENDFOR
	\ENSURE UAV deployment location $\mathbf{v}$ and the GU-UAV connection $\bm{\delta}$.
	\end{algorithmic}
\end{algorithm}

From Algorithm \ref{alg1:K-means}, we obtain the UAV deployment strategy $\mathbf{v}$ and the GU-UAV connection $\bm{\delta}$. Thus, the original problem $\textbf{P0}$ turns into $\textbf{P1}$, which is only related with task offloading decision $\bm{\lambda}$ and resource allocation $\mathbf{f}$:
\begin{equation}
	\begin{split}
		\textbf{P1:}\quad&\min_{\bm{\lambda},\mathbf{f}}\quad \sum_{n=1}^N E_n^{total}+E_h^{total}\\
		&\textrm{s.t.}\quad(\ref{c1})-(\ref{c4}),(\ref{c8}),(\ref{c5}),(\ref{c6}),(\ref{c10}),(\ref{c11}).
	\end{split}
\end{equation}
Note that (\ref{c1}) is the chance constraint, and to deal with it without distribution information and obtain a conservative solution for problem $\textbf{P1}$, we employ DRO to transform (\ref{c1}) into a distributionally robust chance constraint (DRCC) with uncertainty set $\mathcal{P}$. In detail, let $\underset{\mathbb{P}\in\mathcal{P}}{inf}$ denote the lower bound of possibility for all potential distributions, aiming to seek the solution under the worst case. Then, the chance constraint is reformulated as
\begin{equation}
	\underset{\mathbb{P}\in\mathcal{P}}{inf}\quad\mathbf{Pr}_\mathbb{P}\left\{t_m^{total}\leq T_m^{max}\right\}\geq\alpha_m,\forall m\in \mathcal{M}, \label{51}
\end{equation}
which is still complicated due to the random parameter $\Delta_m$ under the uncertainty set $\mathcal{P}$. 

Accordingly, we leverage CVaR mechanism to obtain a conservative estimation for the resource allocation and offloading strategy, which can efficiently improve the reliability while reducing the energy consumption. Generally, CVaR is an indicator to evaluate the risk quantification. It is defined as the conditional expectation value of loss that exceeds a certain probability level under a given probability distribution \cite{Value-Sarykalin}, \cite{Optimization-Tyrrell}. The inherent relationship between the loss function $\phi(\xi)$ for random parameter $\xi$ and CVaR under safety factor $\alpha$ is
\begin{equation}
	\mathbb{P}\left\{\phi(\xi)\leq\mathbb{P}-CVaR_\alpha(\phi(\xi))\right\}\geq\alpha.\label{equation30}
\end{equation}
Then, the CVaR constraint in (\ref{equation30}) can constitute a conservative approximation for the DRCC, i.e.,
\begin{equation}
	\begin{split}
	\underset{\mathbb{P}\in \mathcal{P}}{sup}\quad \mathbb{P}-CVaR_{\alpha}\left( \phi \left( \xi \right) \right) \leq 0,\forall \mathbb{P}\in \mathcal{P}\\
	\Leftrightarrow \underset{\mathbb{P}\in \mathcal{P}}{inf}\quad \mathbb{P}\left\{\phi \left( \xi \right) \leq 0\right\}\geq \alpha ,
	\end{split}
\end{equation}
where $\underset{\mathbb{P}\in \mathcal{P}}{sup}$ is the upper bound under distribution $\mathbb{P}$ \cite{Distributionally-Cui}, \cite{Distributionally-Ling}. Moreover, referring \cite{Distributionally-Ding}, we obtain \textbf{Lemma} \textbf{\ref{lemma1}}.

\begin{lemma}
	For $\Theta\in\mathbb{R}$ and $\theta^0\in\mathbb{R}$, if the loss function is $\phi(\xi)=\Theta\xi+\theta^0$, the worst-case CVaR $\underset{\mathbb{P}\in\mathcal{P}}{sup}\quad\mathbb{P}-CVaR_{\alpha}(\phi(\xi))$ can be derived as a second order cone programming, i.e.,\label{lemma1}
	\begin{equation}
		\begin{split}
			\underset{\beta,e,q,z,s}{inf}&\beta +\frac{1}{1-\alpha}\left( e+s \right) ,\\
			&e-\theta ^0+\beta +q-\Theta \mu -z>0,\\
			&e\geq 0,z> 0,\\
			&\begin{Vmatrix}
					q\\
					\Theta\sigma\\
					z-s
			\end{Vmatrix} \leq z+s,\\
		\end{split}
	\end{equation}
in which $\beta$, $e$, $q$, $z$, and $s$ are auxiliary variables. $\mu$ and $\sigma$ are the mean and standard deviation of random parameter $\xi$, respectively.   
\end{lemma}
\begin{proof}
	The detailed proof is in Appendix A.
\end{proof}
Hence, the DRCC in (\ref{51}) can be approximated by a conservative and convex programming problem \cite{Distributionally-Zymler}. Specifically, the complete expression for $t_m^{total}$ is
\begin{equation}
	\begin{split}
		t_m^{total}&=\frac{L_m}{B_u\log_2\left(1+\frac{p_u(\bar{g}_{m}^u+\Delta_{m})}{n_0B_u}\right)}\\
		&+\sum_{n=1}^N\frac{\lambda_m^nL_m}{B_h\log_2\left(1+\frac{p_h g_{n}^hL_sL_l}{k_BT_0B_h}\right)}+\frac{c_mL_m}{f_{m}}.	
	\end{split}
\end{equation}
Since the estimation error $\Delta_m$ is much smaller than the theoretical value of channel gain \cite{Energy-Wu}, we adopt the first-order Taylor expansion to approximate the latency $t_m^{total}$, i.e.,
\begin{equation}\label{equation27}
	\begin{split}
	t_m^{total}\approx &\frac{L_m}{B_u\log_2\left(1+\frac{p_u\bar{g}_{m}^u}{n_0B_u}\right)}+\frac{c_mL_m}{f_{m}}\\
	&+\sum_{n=1}^N\frac{\lambda_m^nL_m}{B_h\log_2\left(1+\frac{p_h g_{n}^hL_sL_l}{k_BT_0B_h}\right)}\\
	&-\frac{L_m\ln 2}{B_u}\frac{p_u\Delta_{m}}{\left(n_0B_u+p_u\bar{g}_{m}^u\right)\ln^2 \left(1+\frac{p_u\bar{g}_{m}^u}{n_0B_u}\right)}.
	\end{split}
\end{equation}
Consequently, the DRCC in (\ref{51}) is reformulated into
\begin{equation}\label{equation28}
	\underset{\mathbb{P}\in\mathcal{P}}{inf}\quad\mathbf{Pr}_\mathbb{P}\left\{\Theta_m\Delta_m+\theta^0_m\leq 0\right\}\geq\alpha_m,\forall m\in \mathcal{M},
\end{equation}
where
\begin{equation}
	\Theta_m=-\frac{L_mf_m\ln 2}{B_u}\frac{p_u}{(n_0B_u+p_u\bar{g}_{m}^u)\ln^2 \left(1+\frac{p_u\bar{g}_{m}^u}{n_0B_u}\right)},
\end{equation}
and
\begin{equation}
	\begin{split}
		\theta^0_m=&\frac{L_mf_m}{B_u\log_2\left(1+\frac{p_u\bar{g}_{m}^u}{n_0B_u}\right)}+c_mL_m\\
		&+\sum_{n=1}^N\frac{\lambda_m^nL_mf_m}{B_h\log_2\left(1+\frac{p_h g_{n}^hL_sL_l}{k_BT_0B_h}\right)}-T_{m}^{max}f_m.	
	\end{split}
\end{equation}
Therefore, according to \textbf{Lemma} \textbf{\ref{lemma1}}, the DRCC in (\ref{51}) with random parameter $\Delta_m$ is reformulated into an MISOCP, i.e.,
\begin{equation}\label{56}
	\begin{split}
		\underset{\beta_m,e_m,q_m,z_m,s_m}{inf}&\beta_m +\frac{1}{1-\alpha_m}\left( e_m+s_m \right)\leq 0,\\
		&e_m-\theta^0_m+\beta_m+q_m-\Theta_m \mu_{m} -z_m>0,\\
		&e_m\geq 0,z_m> 0,\\
		&\begin{Vmatrix}
				q_m\\
				\Theta_m\sigma_{m}\\
				z_m-s_m
		\end{Vmatrix} \leq z_m+s_m,\\
	\end{split}
\end{equation}
where $\beta_m$, $e_m$, $q_m$, $z_m$, and $s_m$ are all auxiliary variables.
Recall that $\mu_{m}$ is the mean value of CSI estimation error $\Delta_m$, and $\sigma_{m}$ is the corresponding standard deviation. Thus, via CVaR, $\textbf{P1}$ is reformulated as
\begin{equation}
	\begin{split}
		\textbf{P2:}\quad&\min_{\bm{\lambda},\mathbf{f},\bm{\beta},\mathbf{e},\mathbf{q},\mathbf{z},\mathbf{s}} \sum_{n=1}^N E_n^{total}+E_h^{total}\\
		&\textrm{s.t.}\quad(\ref{c2})-(\ref{c4}),(\ref{c8}),(\ref{c5}),(\ref{c6}),(\ref{c10}),(\ref{c11}),(\ref{56}),
	\end{split}
\end{equation}
with the MISOCP constraint in (\ref{56}) under the worst-case scenario, and a conservative solution can be obtained to enhance the robustness against the fluctuations. $\bm{\beta}$, $\mathbf{e}$, $\mathbf{q}$, $\mathbf{z}$ and $\mathbf{s}$ are the corresponding vectors of $\beta_m$, $e_m$, $q_m$, $z_m$ and $s_m$, respectively. However, it is still an MIP problem and complicated to deal with both the binary variables and continuous variables.
\subsection{Primal Decomposition for P2}\label{subsection-7}
It is noting that the constraints of problem $\textbf{P2}$ can be divided into constraints $(\ref{c3})$, $(\ref{c4})$, $(\ref{c11})$, $(\ref{c5})$, $(\ref{c6})$, $(\ref{56})$ with respect to variable $\mathbf{f}$, $\bm{\beta}$, $\mathbf{e}$, $\mathbf{q}$, $\mathbf{z}$ and $\mathbf{s}$, as well as constraints $(\ref{c2})-(\ref{c4})$, $(\ref{c10})$, $(\ref{c8})$, $(\ref{c5})$, $(\ref{c6})$, $(\ref{56})$ related with the binary offloading strategies $\bm{\lambda}$ \cite{Tutorial-Palomar}. Specifically, when $\bm{\lambda}$ is fixed, the subproblem \textbf{P3} related to $\mathbf{f}$, $\bm{\beta}$, $\mathbf{e}$, $\mathbf{q}$, $\mathbf{z}$ and $\mathbf{s}$ is accordingly obtained:
\begin{equation}
	\begin{split}
		\textbf{P3:}\quad&\min_{\mathbf{f},\bm{\beta},\mathbf{e},\mathbf{q},\mathbf{z},\mathbf{s}} \sum_{n=1}^N E_n^{total}+E_h^{total}\\
		&\textrm{s.t.}\quad(\ref{c3}),(\ref{c4}),(\ref{c11}),(\ref{c5}),(\ref{c6}),(\ref{56}).
	\end{split}
\end{equation}
In the form of SOCP, $\textbf{P3}$ can be solved by a standard convex optimization toolkit such as CVX.

With the value of $\mathbf{f}$, $\bm{\beta}$, $\mathbf{e}$, $\mathbf{q}$, $\mathbf{z}$ and $\mathbf{s}$, the offloading decision subproblem $\textbf{P4}$ is only related with variable $\bm{\lambda}$:
\begin{equation}
	\begin{split}
		\textbf{P4:}\quad&\min_{\bm{\lambda}}\sum_{n=1}^N E_n^{total}+E_h^{total}\\
		&\textrm{s.t.}\quad(\ref{c2})-(\ref{c4}),(\ref{c10}),(\ref{c8}),(\ref{c5}),(\ref{c6}),(\ref{56}).		
	\end{split}
\end{equation}
As a result, $\textbf{P2}$ can be handled by iteratively solving $\textbf{P3}$ and $\textbf{P4}$. However, $\textbf{P4}$ is still intractable to directly solve due to the binary variables.

\subsection{BWOA for P4}\label{subsection-8}
As for $\textbf{P4}$, the exhaustive search can obtain the optimal solutions. However, as the scale of problem increases, it faces exponential complexity. Hence, we design a meta-heuristic BWOA for efficient solutions, in which each searching agent represents a potential solution for the binary problem. However, since the original BWOA is designed for unconstrained optimization problems, the solution provided by the agent may be not feasible. To address this issue, we adopt the penalty mechanism and reformulate the objective function of $\textbf{P4}$ into $\varGamma (\bm{\lambda})$, which includes both the objective function as well as the penalty value \cite{Whale-Pham}. Specifically, the agents violating constraints are assigned with a higher fitness value under the influence of penalty factors. As such, the constrained problem is effectively transformed without constraints. Based on above discussions, the fitness function related to $\bm{\lambda}$ is defined as:
\begin{equation}\label{fitness-lambda}
	\begin{split}
		\varGamma (\bm{\lambda})&=\sum_{n=1}^N E_n^{total}+E_h^{total}\\
		&+\sum_{m=1}^M\sum_{n=1}^N \vartheta H_{mn,1}(h_{mn,1}(\bm{\lambda}))h_{mn,1}^2(\bm{\lambda})\\
		&+\sum_{n=1}^N \vartheta H_{n,2}(h_{n,2}(\bm{\lambda}))h_{n,2}^2(\bm{\lambda})+\vartheta H_{3}(h_{3}(\bm{\lambda}))h_{3}^2(\bm{\lambda})\\
		&+\sum_{n=1}^N \vartheta H_{n,4}(h_{n,4}(\bm{\lambda}))h_{n,4}^2(\bm{\lambda})+\vartheta H_{5}(h_{5}(\bm{\lambda}))h_{5}^2(\bm{\lambda})\\
		&+\vartheta H_{6}(h_{6}(\bm{\lambda}))h_{6}^2(\bm{\lambda})+\sum_{m=1}^M \vartheta H_{m,7}(h_{m,7}(\bm{\lambda}))h_{m,7}^2(\bm{\lambda}),
	\end{split}
\end{equation}
where the penalty factor $\vartheta$ is set as $10^{5}$. $H(\cdot)$ is an index function. $H(h(\bm{\lambda}))=0$ if $h(\bm{\lambda})\leq 0$, and otherwise $H(h(\bm{\lambda}))=1$. Hence, by introducing the penalty factor and index function, the solution which violates the constraints leads to an increasing fitness. The penalty factors act as a role to prevent agents from searching infeasible solutions during their explorations and determine whether the current solution satisfies the corresponding constraints. $h(\bm{\lambda})$ is defined based on the constraints of $\textbf{P4}$, i.e., 
\begin{equation}
	\begin{cases}
		h_{mn,1}(\bm{\lambda})=\lambda_m^n-\delta_m^n,\\
		h_{n,2}(\bm{\lambda})=\sum\limits_{m=1}^M\lambda_m^nE_m^h+\sum\limits_{m=1}^M(\delta_m^n-\lambda_m^n)E_m^{cu}-E_n^{max},\\
		h_{3}(\bm{\lambda})=\sum\limits_{n=1}^N\sum\limits_{m=1}^M\lambda_m^nE_m^{ch}-E_h^{max},\\
		h_{n,4}(\bm{\lambda})=\sum\limits_{m=1}^M(\delta_m^n-\lambda^n_m)f_m-F_{max}^n,\\
		h_{5}(\bm{\lambda})=\sum\limits_{m=1}^M \sum\limits_{n=1}^N \lambda_m^n f_{m}-F_{max}^h,\\
		h_{6}(\bm{\lambda})=\sum\limits_{m=1}^M\sum\limits_{n=1}^N\lambda_m^n-H,\\
		h_{m,7}(\bm{\lambda})=-e_m+\theta^0_m-\beta_m-q_m+\Theta_m \mu_{m}+z_m.\\
	\end{cases}
\end{equation}
Consequently, the objective function of BWOA to tackle $\textbf{P4}$ is further transformed as
\begin{equation}\label{fitness}
	\begin{split}
		\textbf{P5:}\quad&\min_{\bm{\lambda}}\varGamma (\bm{\lambda})\\
		\textrm{s.t.}&\quad(\ref{c10}).
	\end{split}
\end{equation}
Based on the objective function of $\textbf{P5}$, BWOA is further enhanced to search for quasi-optimal solutions, which is a swarm-based technology in light of hunting behaviors of whales \cite{Binary-Eid}, \cite{Binary-Kumar}. The positions of agents are updated iteratively based on the social behaviors of whales including exploration and exploitation \cite{Whale-Mirjalili}. The quality of each solution is evaluated by calculating the fitness function $\varGamma(\bm{\lambda})$ in (\ref{fitness-lambda}). The position of each agent $X(i_2)$ during iteration $i_2$ represents a potential solution for $\bm{\lambda}$. Specifically, the procedures of BWOA include encircling prey, spiral updating, and searching for prey, detailed as follows.
\subsubsection{Encircling Prey}For the agent encircling prey, it might update its current position linearly toward the current optimal solution $\vec{X}^*(i_2)$. The position of the agent in the next integration is
\begin{equation}\label{update-X1}
	\vec{X}(i_2+1)=\begin{cases}
		\complement(\vec{X}(i_2)),\quad&\text{if}\quad P_{BWOA}<\tau_{ep},\\
		\vec{X}(i_2),\quad&\text{if}\quad P_{BWOA}\geq\tau_{ep},
	\end{cases}
\end{equation}
where $\tau_{ep}$ is the step size, which is a possibility to determine whether there is a switch (from $0$ to $1$ or from $1$ to $0$) between the current bit value and the value in the next iteration. Specifically,
\begin{equation}\label{update-ep}
	\tau_{ep}=\frac{1}{1+\exp \left( -10\left( \vec{A}\cdot \vec{D}-0.5 \right) \right)},
\end{equation}
where 
\begin{equation}\label{update-A}
	\vec{A}=2\vec{a}\cdot\vec{r_1}-\vec{a},
\end{equation}
and
\begin{equation}\label{update-C}
	\vec{C}=2\cdot\vec{r_2},
\end{equation}
are coefficient factors. The operation $\cdot$ indicates element-wise multiplication \cite{Whale-Pham}. $\vec{a}$ linearly decreases from $2$ to $0$ during the iteration:
\begin{equation}\label{update-a}
	a=2-i_2\times\frac{2}{I_2^{max}},
\end{equation}
where $i_2$ is the index of current iteration, $I_2^{max}$ is the maximum number of iterations. $\vec{r_1}$ and $\vec{r_2}$ are random vectors within $[0,1]$. $\complement(\cdot)$ represents the complement operation. $P_{BWOA}\in [0,1]$ is a basis for action selection in the mechanism of encircling prey, and the distance vector $\vec{D}$ is calculated by
\begin{equation}\label{update-D}
	\vec{D}=|\vec{C}\cdot \vec{X}^*\left( i_2 \right) -\vec{X}\left( i_2 \right) |.
\end{equation}
\subsubsection{Spiral Updating}The agent tends to approach the current optimal individual in either encircling prey or spiral manner. In detail, the spiral updating mechanism for position updating in BWOA is
\begin{equation}\label{update-Xsu}
	\vec{X}(i_2+1)=\begin{cases}
		\complement(\vec{X}(i_2)),\quad&\text{if}\quad P_{BWOA}<\tau_{su},\\
		\vec{X}(i_2),\quad&\text{if}\quad P_{BWOA}\geq\tau_{su}.
	\end{cases}
\end{equation}
Moreover, the step size $\tau_{su}$ is calculated as
\begin{equation}\label{update-su}
	\tau _{su}=\frac{1}{1+\exp \left( -10\left( \vec{A}\cdot \vec{D^\prime}-0.5 \right) \right)},
\end{equation}
where $\vec{A}$ is updated based on (\ref{update-A}), and $\vec{D^\prime}$ is updated as
\begin{equation}\label{update-D'}
	\vec{D^\prime}=|\vec{X}^*(i_2)-\vec{X}(i_2)|.
\end{equation}
\begin{algorithm}[t]
	\caption{BWOA for the Offloading Decision\label{alg2:BWOA}}
	\begin{algorithmic}[1]
	\STATE \textit{Initialization:} Set the iteration index $i_2=1$, maximum number of iteration $I_2^{max}$ and the agent population $X_{k}$, $k={1,2,\ldots,K}$.
	\STATE {Calculate the fitness value of the search agents and obtain the best search agent $\vec{X}^*(i_2)$.}
	\REPEAT
		\FOR{$k\in\{1,2,\cdots,K\}$}
			\STATE {Update $A$, $C$, and $a$ according to (\ref{update-A}), (\ref{update-C}) and (\ref{update-a}), respectively.\label{line5}}
			\STATE {Generate the parameter $P_{rand}\in[0,1]$ randomly.}
			\IF{$P_{rand}\geq 0.5$}
				\STATE {Update $\vec{D}$ by (\ref{update-D'}) and $\tau_{su}$ by (\ref{update-su}).\label{line7}}
				\STATE {Update the position $\vec{X}(i_2)$ by (\ref{update-Xsu}).\label{line8}}
			\ELSE
				\IF {$|A|\geq 1$}
					\STATE {Select a random agent $\vec{X}_{rand}$ and update $\vec{D}$ based on (\ref{update-D''}).\label{line11}}
					\STATE {Update $\tau_{sp}$ via (\ref{update-sp}) and $\vec{X}(i_2)$ via (\ref{update-Xsp}).\label{line12}}
				\ELSE
					\STATE {Update $\vec{D}$ via (\ref{update-D}) and $\tau_{ep}$ via (\ref{update-ep}).\label{line14}}
					\STATE {Update the positions of agents $\vec{X}(i_2)$ via (\ref{update-X1}).\label{line15}}
				\ENDIF
			\ENDIF
		\ENDFOR
		\STATE {Calculate the fitness value of each agent via (\ref{fitness-lambda}) and obtain the best position $\vec{X}^*(i_2)$ of the agents.\label{line19}}
		\STATE {Update the iteration index $i_2=i_2+1$.}
	\UNTIL {$i_2> I_2^{max}$}
	\ENSURE{The best fitness value $\varGamma$ and offloading decision $\bm{\lambda}$.}
	\end{algorithmic}
\end{algorithm}
\subsubsection{Search for Prey}To achieve more exploration and avoid falling into local optima, some agents conduct random searches instead of updating toward the current optimal. This mechanism is called $\textit{search for prey}$ in BWOA. Since the exploration may bring agents to deviate from the current optima, it enhances the global search capacities of the agents. The updating rule for the positions of agents is
\begin{equation}\label{update-Xsp}
	\vec{X}(i_2+1)=\begin{cases}
		\complement(\vec{X}(i_2)),\quad&\text{if}\quad P_{BWOA}<\tau_{sp},\\
		\vec{X}(i_2),\quad&\text{if}\quad P_{BWOA}\geq\tau_{sp},
	\end{cases}
\end{equation}
where
\begin{equation}\label{update-sp}
	\tau_{sp}=\frac{1}{1+\exp \left( -10\left( \vec{A}\cdot \vec{D^{\prime\prime}}-0.5 \right) \right)}.
\end{equation}
Furthermore, $\vec{A}$ is obtained by (\ref{update-A}), and $\vec{D^{\prime\prime}}$ is calculated as
\begin{equation}\label{update-D''}
	\vec{D^{\prime\prime}}=|\vec{C}\cdot \vec{X}_{rand}\left( i_2 \right) -\vec{X}\left( i_2 \right) |,
\end{equation}
where $\vec{X}_{rand}$ denotes the position of a random selected agent.

\begin{table}[!t]
	\renewcommand\arraystretch{1.25}
	\begin{center}
	\caption{PARAMETER SETTING}\label{parameter}
	\begin{tabular}{|c|c||c|c|}
	\hline
	Parameter & Value & Parameter & Value \\
	\hline
	\hline
	$p_u$ & $0.5W$  & $p_h$ & $2W$ \\
	\hline
	$g_0^u$ & $-50dB$ & $n_0$ & $-174dBm/Hz$\\
	\hline
	$B_u$ & $5MHz$ & $B_h$ & $5MHz$\\
	\hline
	$T_0$ & $1000K$ & $k_B$ & $1.38\times 10^{-23}J/K$\\
	\hline
	$f_c$ & $2.4GHz$ & $v_c$ & $3\times 10^8 m/s$\\
	\hline
	$L_l$ & $-23dB$ & $g_m^h$ & $42dB$\\
	\hline
	$c_m$ & $300cycles/bit$ & $T^{max}_m$ & $20s$\\
	\hline
	$H$ & $10$ & $\alpha_m$ & $95\%$\\
	\hline
	$F_{max}^n$ & $8\times 10^{9} Hz$ & $F_{max}^h$ & $4\times 10^{11} Hz$\\
	\hline
	$\varepsilon_n$ & $10^{-27}$ & $\varepsilon_h$ & $10^{-28}$\\
	\hline
	$E^{max}_n$ & $200J$ & $E^{max}_h$ & $20kJ$\\
	\hline
	$\mu_m$ & 0 & $\sigma_m$ & $0.1\bar{g}_{m}^u$\\
	\hline
	$\varsigma_1$ & 0.4 & $\varsigma_2$ & 0.2\\
	\hline
	\end{tabular}
	\end{center}
\end{table}
\subsubsection{Algorithm Design}
The BWOA is provided in Algorithm \ref{alg2:BWOA} for solving $\textbf{P5}$. To begin with, $K$ agents are set randomly. Then, the positions of agents are initialized, and the value of fitness function $\varGamma (\bm{\lambda})$ for $\bm{\lambda}$ can be calculated according to (\ref{fitness-lambda}). Furthermore, at each iteration, the corresponding parameters $a$, $A$ and $C$ are updated (line \ref{line5}). Each agent adopts different strategies based on the probability and updates its position in the next iteration. Specifically, each agent has a possibility of $0.5$ to update the parameters and its positions towards optimal in the spiral manner according to (\ref{update-Xsu}) (lines \ref{line7}-\ref{line8}). Otherwise, if the parameter $|A|\geq 1$, the agent is expected to randomly select another agent and search for prey to its direction (lines \ref{line11}-\ref{line12}). Then, the position in the next iteration is obtained by (\ref{update-Xsp}). When $|A|<1$, the agent encircling prey and linearly approach the individual with best fitness value via (\ref{update-X1}) (lines \ref{line14}-\ref{line15}). After all agents updating their positions, their fitness function values are calculated and the position of the current best agent is updated (line \ref{line19}). The above process is repeated until the result converges \cite{Tutorial-Palomar}. Finally, the offloading decisions concerning $\bm{\lambda}$ can be derived. The implementation of Algorithm \ref{alg2:BWOA} is mainly related with the number and dimension of agents, i.e., $K$ and $MN$, respectively. Moreover, the $MN+M+2N+3$ constraints of problem $\textbf{P4}$ impact the computational complexity to calculate the index functions. Hence, the complexity of Algorithm \ref{alg2:BWOA} is $\mathcal{O}(KMN(MN+M+2N+3)I_2^{max})$.

\section{Simulation Results}\label{sec5}

\begin{figure}[t]
	\centering
	\subfloat[Distribution of GUs.]{
	\includegraphics[width=0.8\linewidth]{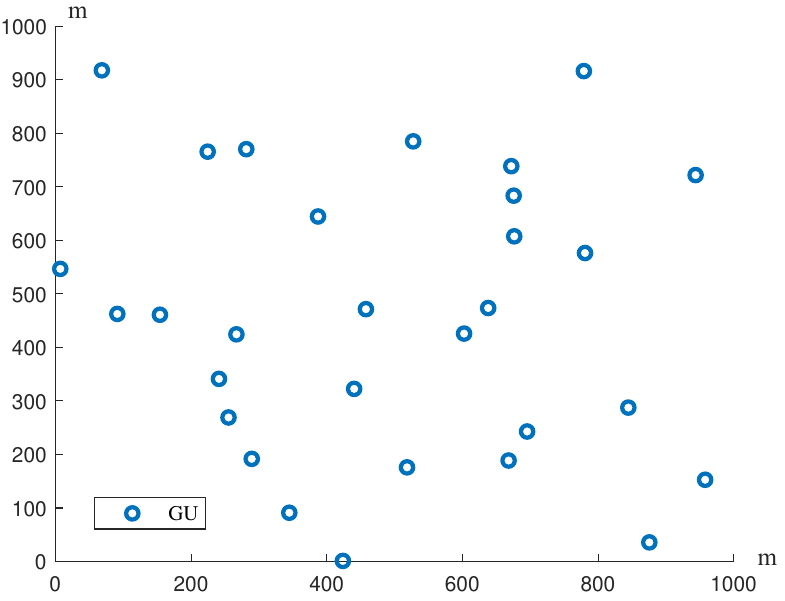}
	}
	\quad
	\subfloat[UAV deployment results.]{
	\includegraphics[width=0.8\linewidth]{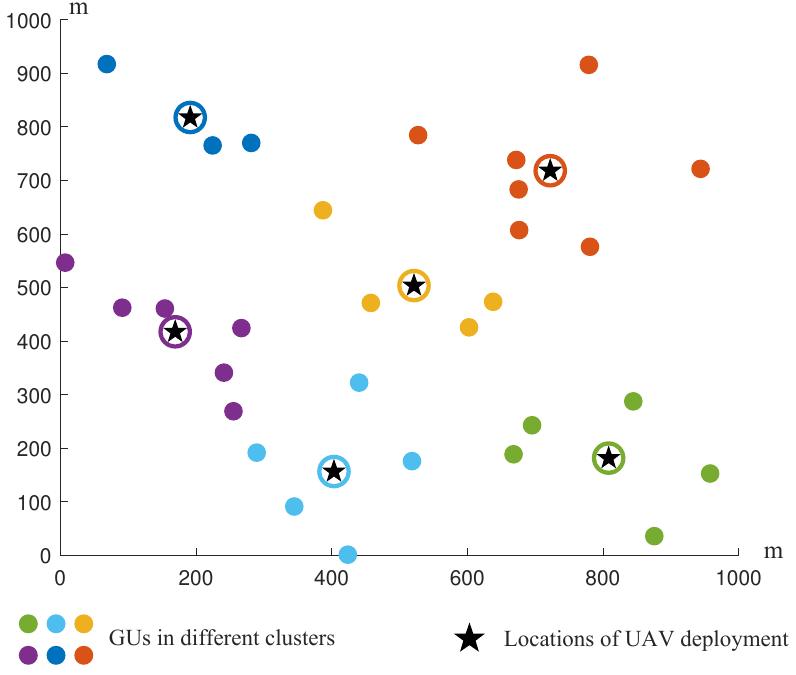}
	}
	\caption{Cluster results via the WKD algorithm (30 GUs and 6 UAVs in the 1km$\times$1km area).}\label{result-distribution}
\end{figure}

\begin{figure}[t]
	\centering
	\subfloat[]{
	\includegraphics[width=0.8\linewidth]{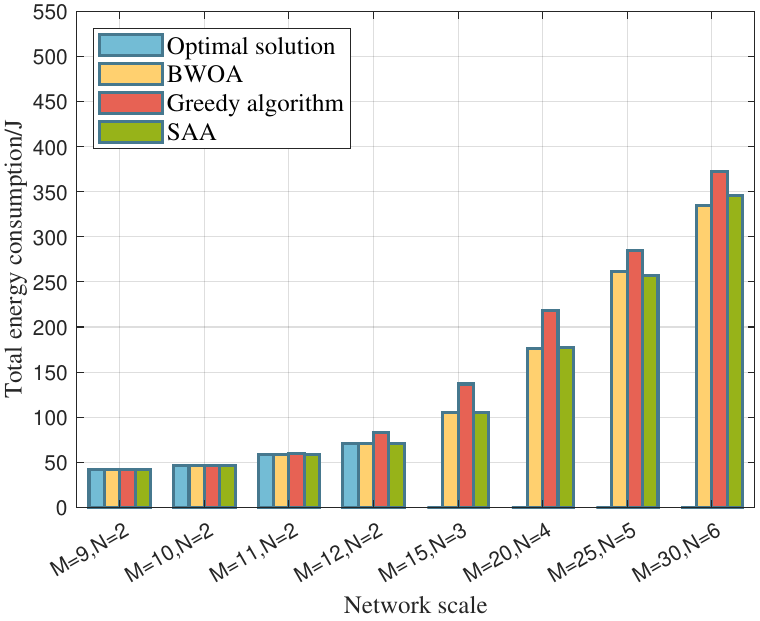}
	}
	\quad
	\subfloat[]{
	\includegraphics[width=0.815\linewidth]{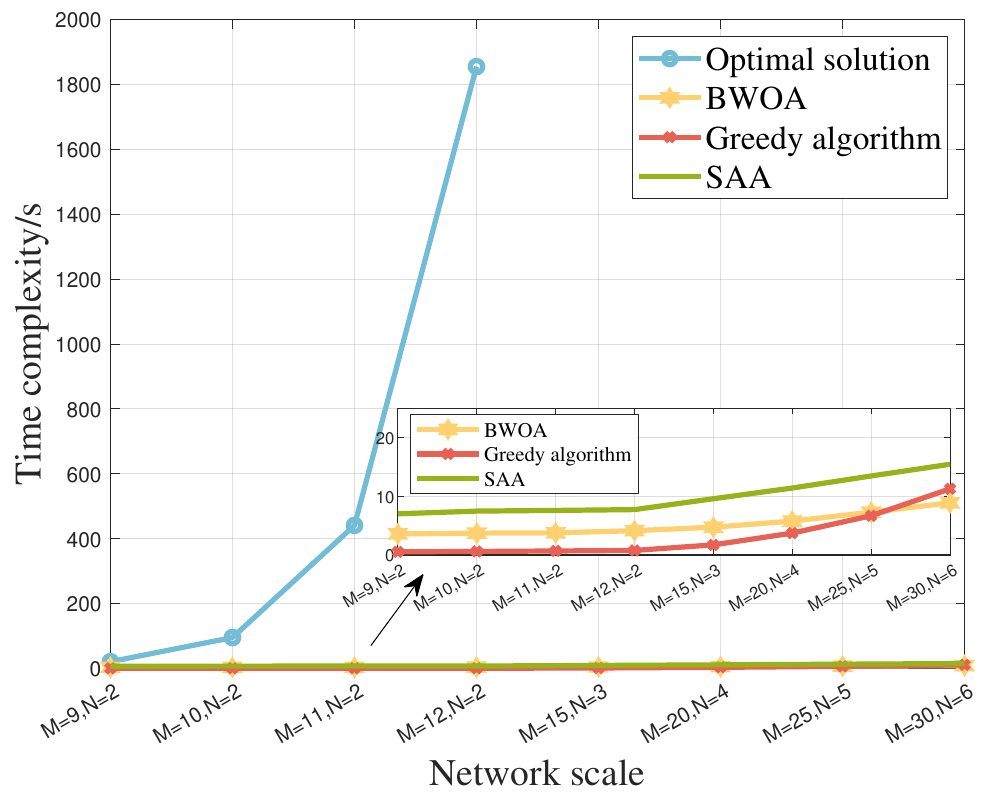}
	}
	\caption{Performance of different algorithms under different network scales.}\label{result-algorithm}
\end{figure}

\begin{figure}[t]
	\centering
	\subfloat[]{
	\includegraphics[width=0.8\linewidth]{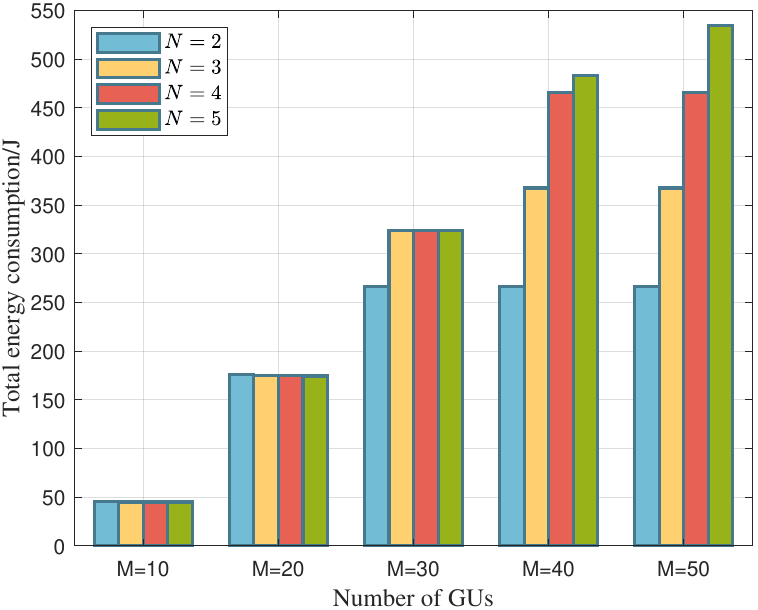}
	}
	\quad
	\subfloat[]{
	\includegraphics[width=0.8\linewidth]{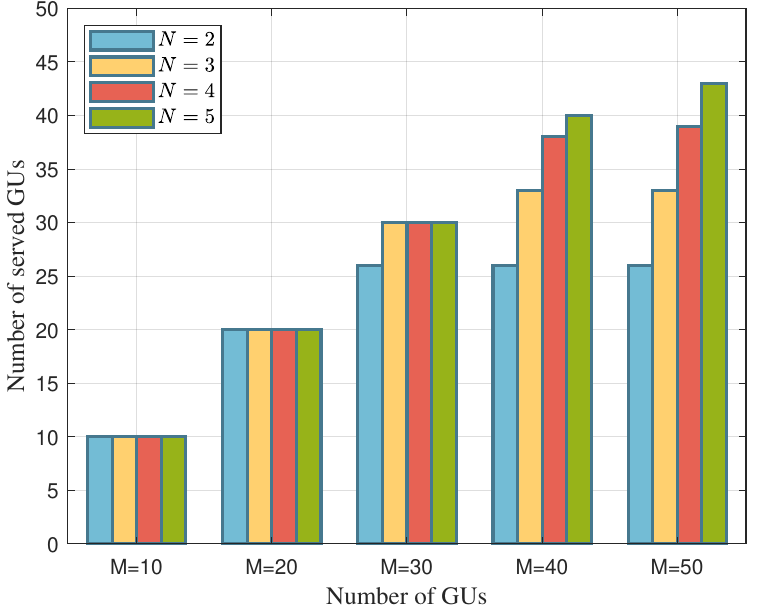}
	}
	\caption{Energy cost and number of served GUs $\emph{v.s.}$ the number of GUs with different UAV scale.}\label{result-servednumber}
\end{figure}

In this section, simulations are conducted to evaluate the proposed algorithms. The GUs are distributed randomly in a $1$km$\times$$1$km area. The altitude of UAVs is $100$m and the coordinate of the HAP is $\varpi_h=[500,500,2\times 10^4]$m. The major parameters are in Table \ref{parameter} \cite{Hierarchical_Jia}, \cite{AI-based_Maryam}. The data size of tasks is within $[50, 70]$Mbits.

The distribution of GUs and UAVs is depicted in Fig. \ref{result-distribution}(a), including 30 GUs and 6 UAVs. Algorithm \ref{alg1:K-means} is applied to obtain the deployment positions of UAVs, and the clustering results of 30 GUs as well as the positions of 6 UAVs are shown in Fig. \ref{result-distribution}(b). It is observed that UAVs are deployed at the centers of GU clusters.

To evaluate the effectiveness and efficiency of the proposed BWOA, it is compared with the optimal solution (obtained by exhaustive search), greedy offloading algorithm, as well as the simulated annealing algorithm (SAA), as shown in Fig. \ref{result-algorithm}. Specifically, from Fig. \ref{result-algorithm}(a), it is observed that the optimization results of BWOA are close to optimal and outperform the greedy algorithm and SAA. Meanwhile, Fig. \ref{result-algorithm}(b) provides the corresponding time complexity. Although the optimal solution is obtained by the exhaustive search, the time cost is not acceptable in the large-scale situations. Besides, the time cost of BWOA is lower than the greedy algorithm and SAA with the network scale increasing. Hence, the BWOA shows near optimal performance with lower time complexity.

Fig. \ref{result-servednumber} depicts the performance of energy cost and the number of served GUs with different number of UAVs. As the number of served GUs increases, more energy is consumed as expected. Moreover, with the same number of served GUs, the energy consumption is almost unchanged by simply increasing the number of UAVs. Besides, note that there exist unserved GUs if the resources are inadequate or the QoS demand of tasks cannot be comprehensively satisfied, indicating that the number of GUs in the network exceeds the network's capacity and the constraints are violated with penalty factors. This is on accounting for the insufficient UAVs which leads to the limited aerial resources and less served GUs. In such a case, the number of served GUs can be improved by deploying more UAVs to alleviate load pressure.

The superiority of proposed WKD based algorithm is verified in Fig. \ref{result-position}. The total energy cost obtained via our proposed algorithms is compared with the results that UAVs are randomly deployed and GUs are randomly connected to the UAVs (R\&R). Note that less energy is consumed compared with the baseline result with the same served GUs. Moreover, compared with the R\&R method, the proposed WKD method can increase the number of GUs that the system can accommodate and make a better use of available resources. Besides, the overloading or under utilization of UAVs can be avoided.

Fig. \ref{result-robust} verifies the robustness of the proposed mechanism (CVaR-based DRCC). As provided in Fig. \ref{result-robust}(a), we compare the results optimized via DRCC and CVaR based method with the "ideal circumstance" to evaluate the impact of imperfect CSI, in which the accurate CSI is obtained. When there exist CSI estimation errors, more energy is consumed compared with the ideal circumstance. This is accounted by the fact that the MEC servers allocate more computing resources for the tasks to cope with the impact of environmental disturbances, as shown in Fig. \ref{result-robust}(b).

The impacts of the tolerable delay of tasks are shown in Fig. \ref{result-latency}. As the maximum tolerable delay increases, the energy consumption decreases, due to the less consumption of computing tasks in Fig. \ref{result-latency}(a), since it provides more time for MEC processing and the required CPU frequency is decreased, as shown in Fig. \ref{result-latency}(b). Hence, the computation energy is decreased. In addition, Fig. \ref{result-GUpower} shows the influence of GU transmission power. It is observed that with the increment of transmission power of GUs, both the total energy consumption and the required CPU frequency are decreased. It is explained that the transmission delay is decreased due to the increment of transmission rate, so the time is sufficient for computation.

The effect of transmission power of UAVs on the performance is discussed in Fig. \ref{result-UAVpower}. As expected in Fig. \ref{result-UAVpower}(a), UAVs consume more energy with the increment of transmission power. Moreover, in Fig. \ref{result-UAVpower}(b), more transmission power of UAVs leads to less CPU frequency consumption. It is explained that the transmission rate for G2U link is growing with the increment of transmission power, and thus, there is more remaining time for MEC processing, resulting in less CPU frequency required.
\begin{figure}
	\centering
	\includegraphics[width=0.865\linewidth]{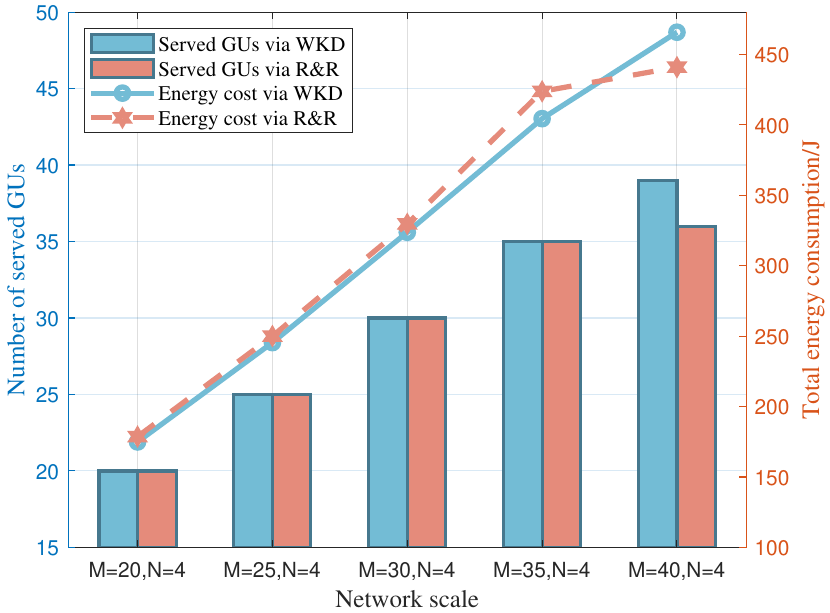}
	\caption{Performance of the algorithm with deployment optimization.}\label{result-position}
\end{figure}

\begin{figure}
	\centering
	\subfloat[]{
	\includegraphics[width=0.8\linewidth]{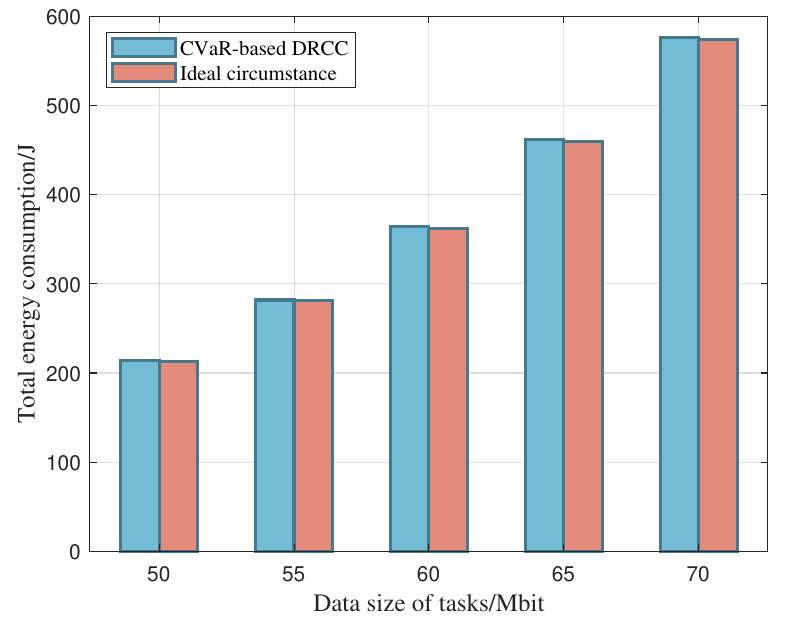}
	}
	\quad
	\subfloat[]{
	\includegraphics[width=0.8\linewidth]{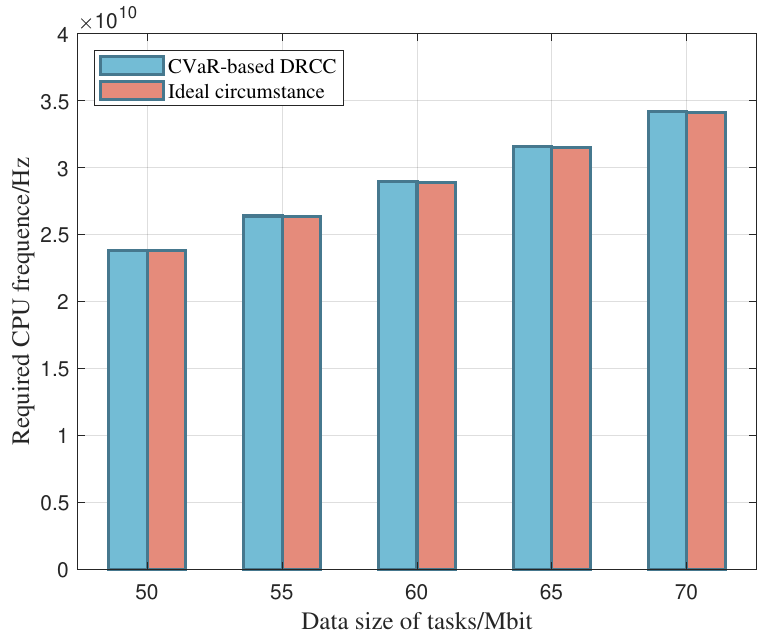}
	}
	\caption{Performance of the model with CSI estimation errors and ideal conditions.}\label{result-robust}
\end{figure}

\begin{figure}
	\centering
	\subfloat[]{
	\includegraphics[width=0.8\linewidth]{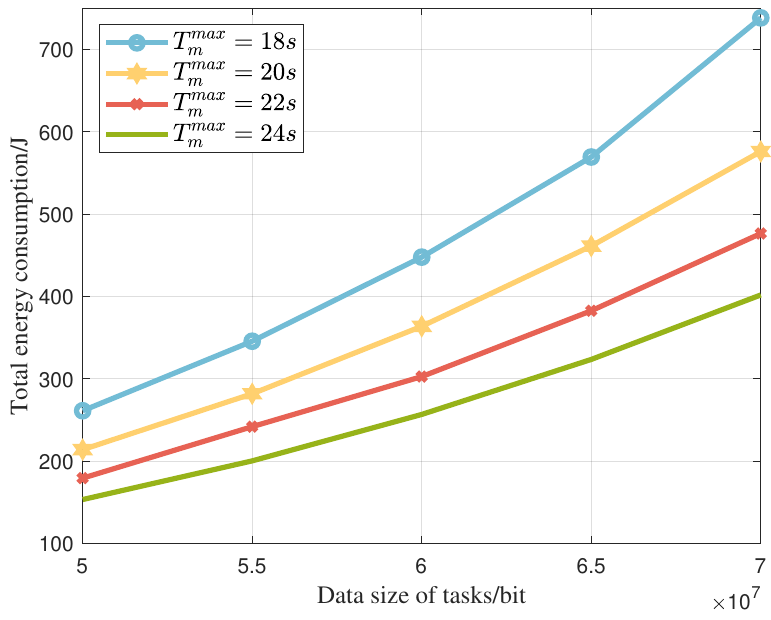}
	}
	\quad
	\subfloat[]{
	\includegraphics[width=0.8\linewidth]{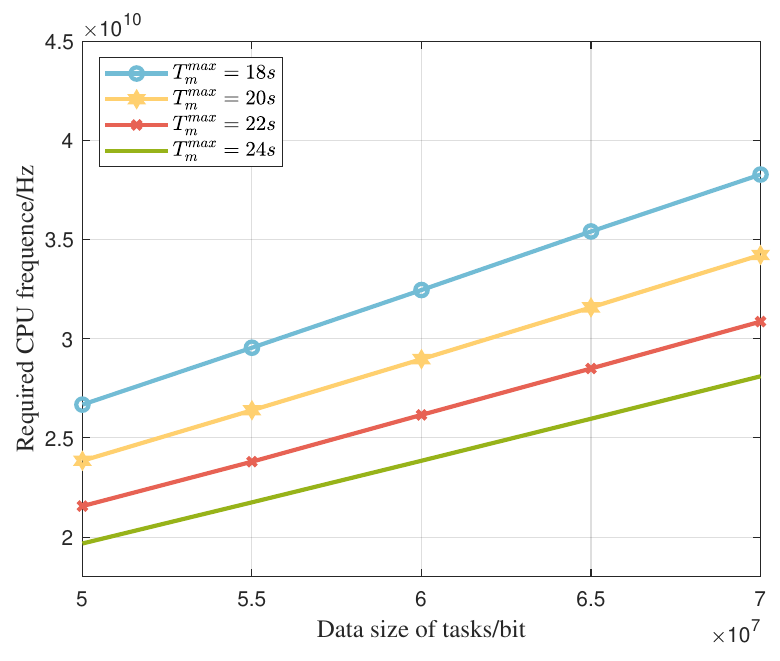}
	}
	\caption{Impact of tolerable delay on the network performance.}\label{result-latency}
\end{figure}

\begin{figure*}
	\centering
	\subfloat[]{
	\includegraphics[width=0.4\linewidth]{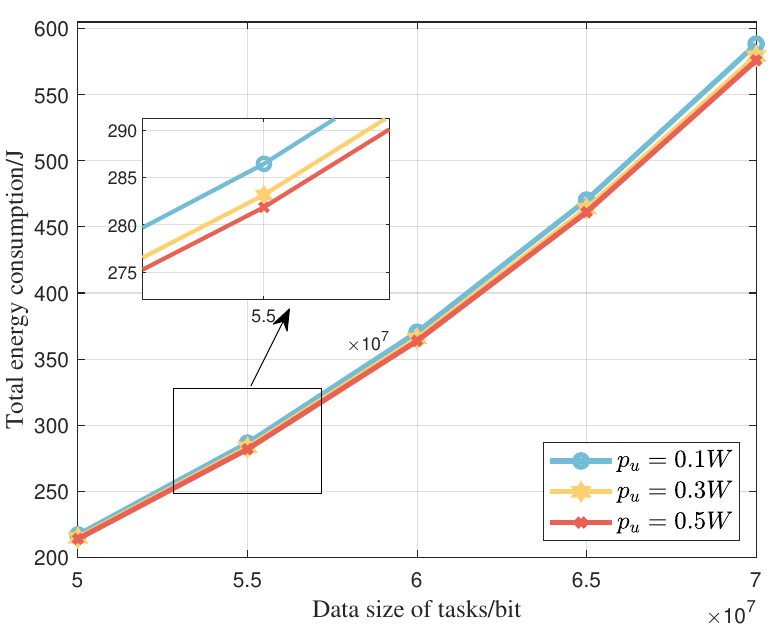}
	}
	\quad
	\subfloat[]{
	\includegraphics[width=0.4\linewidth]{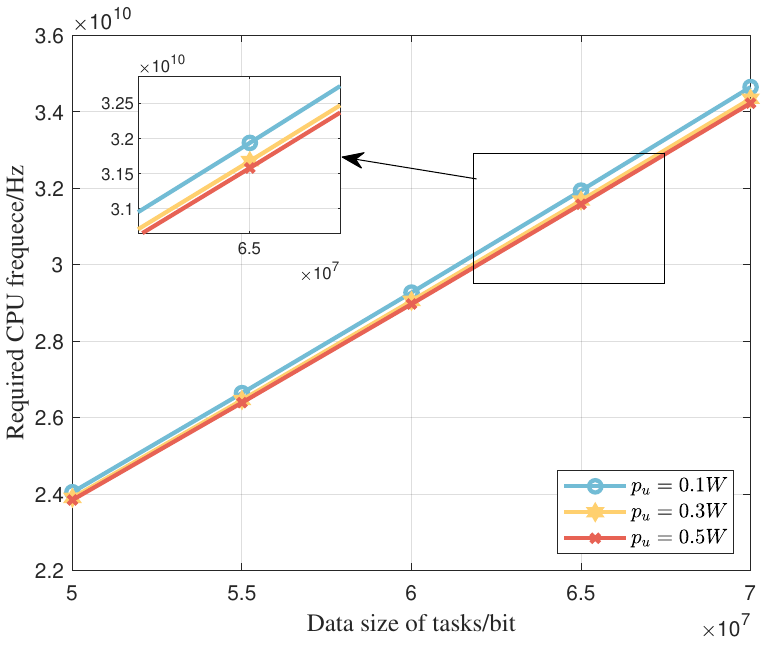}
	}
	\caption{Impact of transmission power of GUs.}\label{result-GUpower}
\end{figure*}
\begin{figure*}
	\centering
	\subfloat[]{
	\includegraphics[width=0.4\linewidth]{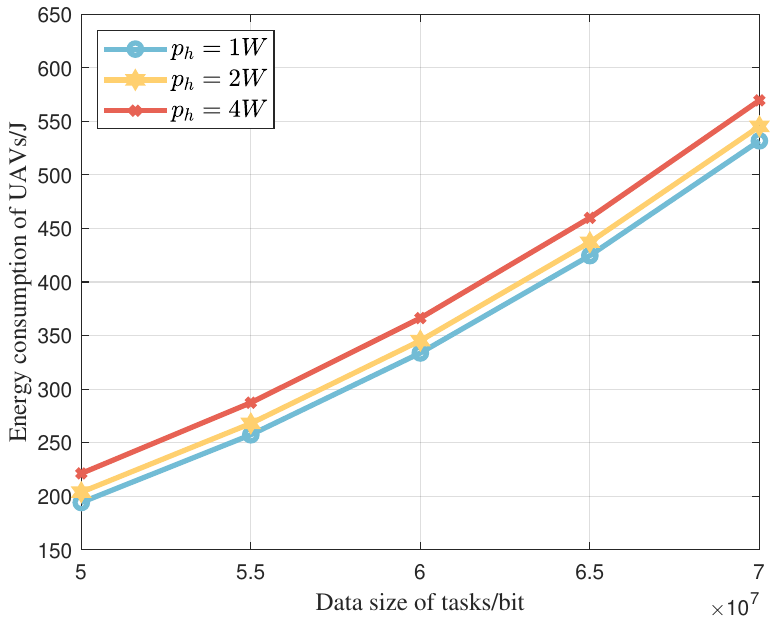}
	}
	\quad
	\subfloat[]{
	\includegraphics[width=0.4\linewidth]{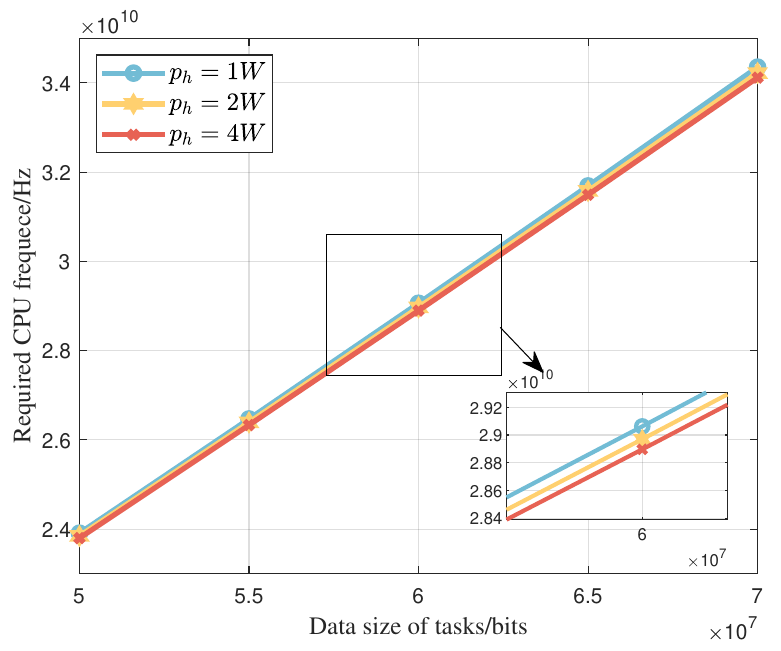}
	}
	\caption{Impact of transmission power of UAVs.}\label{result-UAVpower}
\end{figure*}

\section{Conclusion}\label{sec6}
In this paper, we proposed a hierarchical aerial MEC model consisting of multiple UAVs and an HAP. Considering the limitation of battery capacities, we jointly optimized the UAV deployment strategies, resources allocation and offloading decisions to minimize the total energy consumption. Taking into account the imperfect CSI affected by the unpredictable environmental factors, we established an uncertainty set for CSI estimation errors and formulated the problem with the chance constraint. As for the solution, we designed the WKD based algorithm for the deployment of UAVs. Moreover, the chance constraint was transformed into a DRCC, and accordingly approximated into an MISOCP form under the worst case by employing the CVaR mechanism. Additionally, the MIP problem was further decomposed into two subproblems. To tackle the binary subproblem, BWOA was designed. Finally, we conducted extensive simulations to evaluate the robustness and efficiency. The results showed superiority of the proposed algorithm in the near optimal solution and low time complexity compared with other baseline algorithms. In the future works, the dynamic deployment for UAVs will be further investigated.

\bibliographystyle{IEEEtran}
\bibliography{reference.bib}

\begin{IEEEbiography}[{\includegraphics[width=1in,height=1.25in,clip,keepaspectratio]{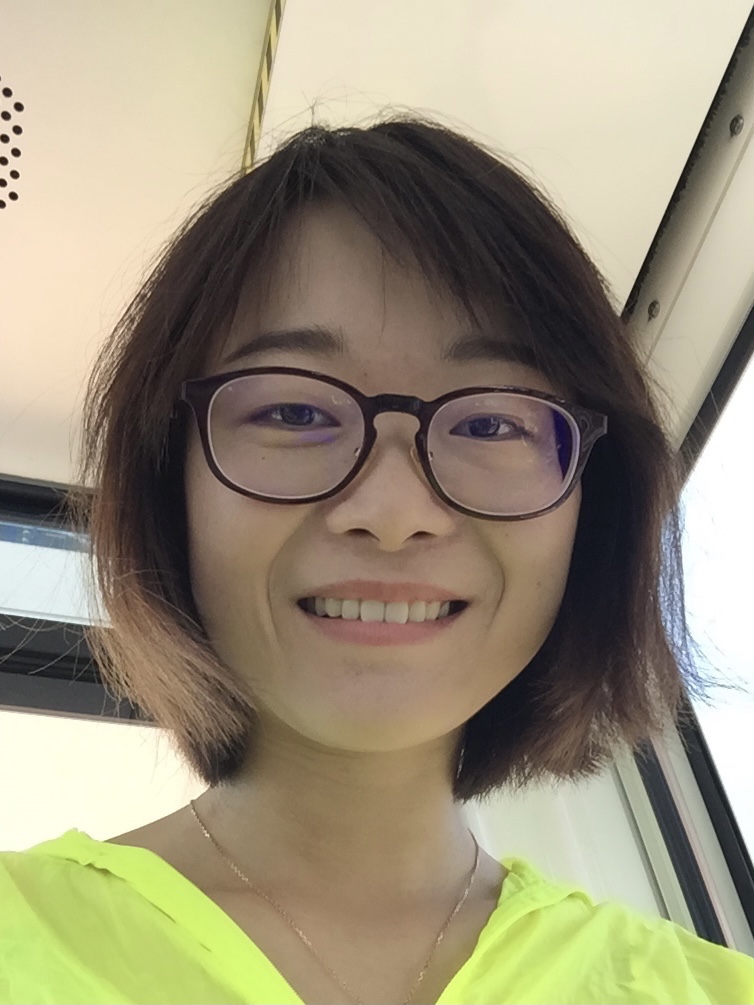}}] {Ziye Jia} (Member, IEEE) received the B.E., M.S., and Ph.D. degrees in communication and information systems from Xidian University, Xi'an, China, in 2012, 2015, and 2021, respectively. From 2018 to 2020, she was a Visiting Ph.D. Student with the Department of Electrical and Computer Engineering, University of Houston. She is currently an Associate Professor with the Key Laboratory of Dynamic Cognitive System of Electromagnetic Spectrum Space, Ministry of Industry and Information Technology, Nanjing University of Aeronautics and Astronautics, Nanjing, China. Her current research interests include space-air-ground networks, aerial access networks, UAV networking, resource optimization,  machine learning, etc.
\end{IEEEbiography}

\begin{IEEEbiography}[{\includegraphics[width=1in,height=1.25in,clip,keepaspectratio]{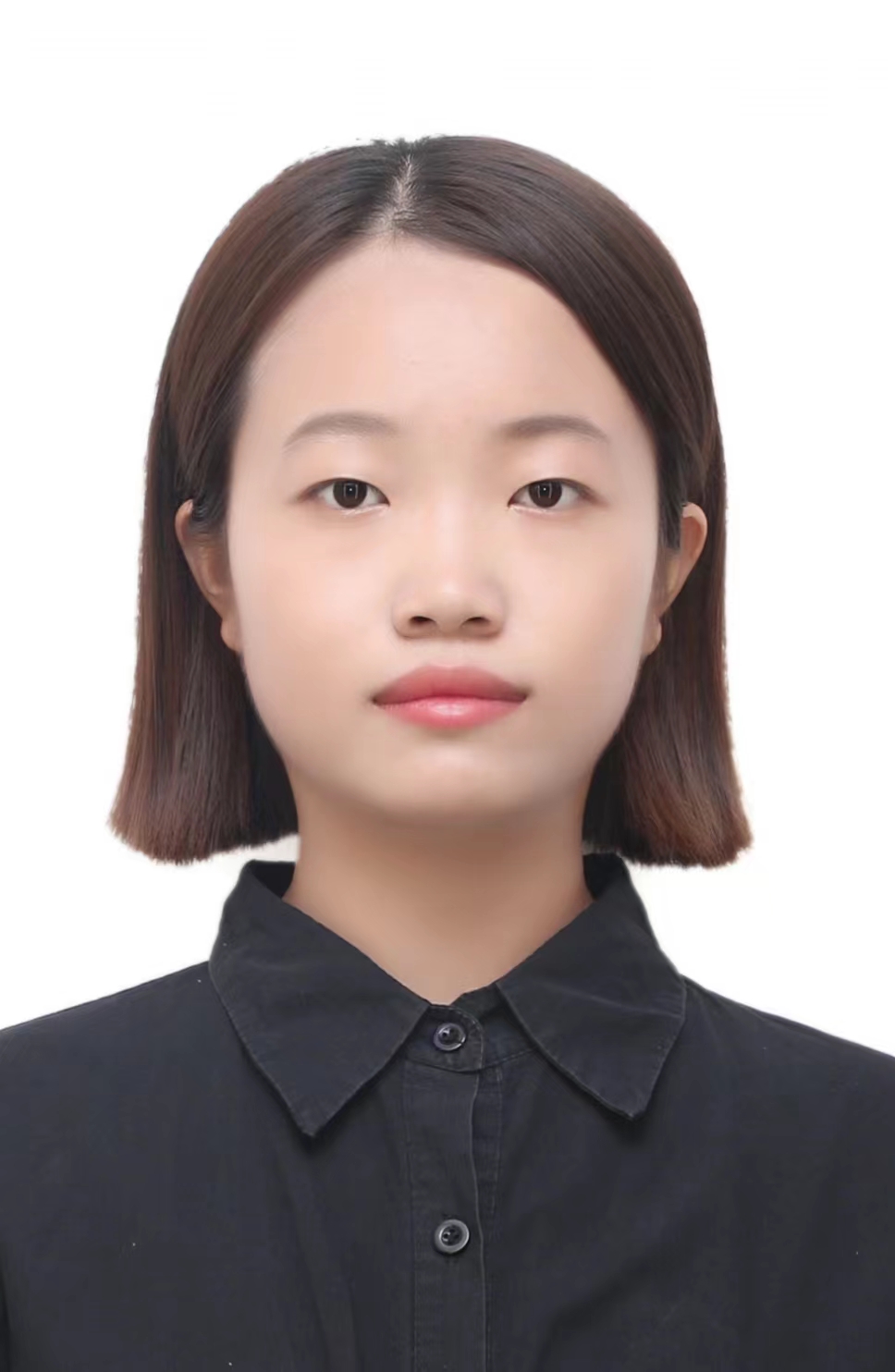}}] {Can Cui} is a postgraduate student with the College of Electronic and Information Engineering, Nanjing University of Aeronautics and Astronautics, Nanjing, China. Her current research interests include convex optimization and its applications in computation offloading and resource allocation, edge computing, and low-altitude intelligent networks.
\end{IEEEbiography}

\begin{IEEEbiography}[{\includegraphics[width=1in,height=1.25in,clip,keepaspectratio]{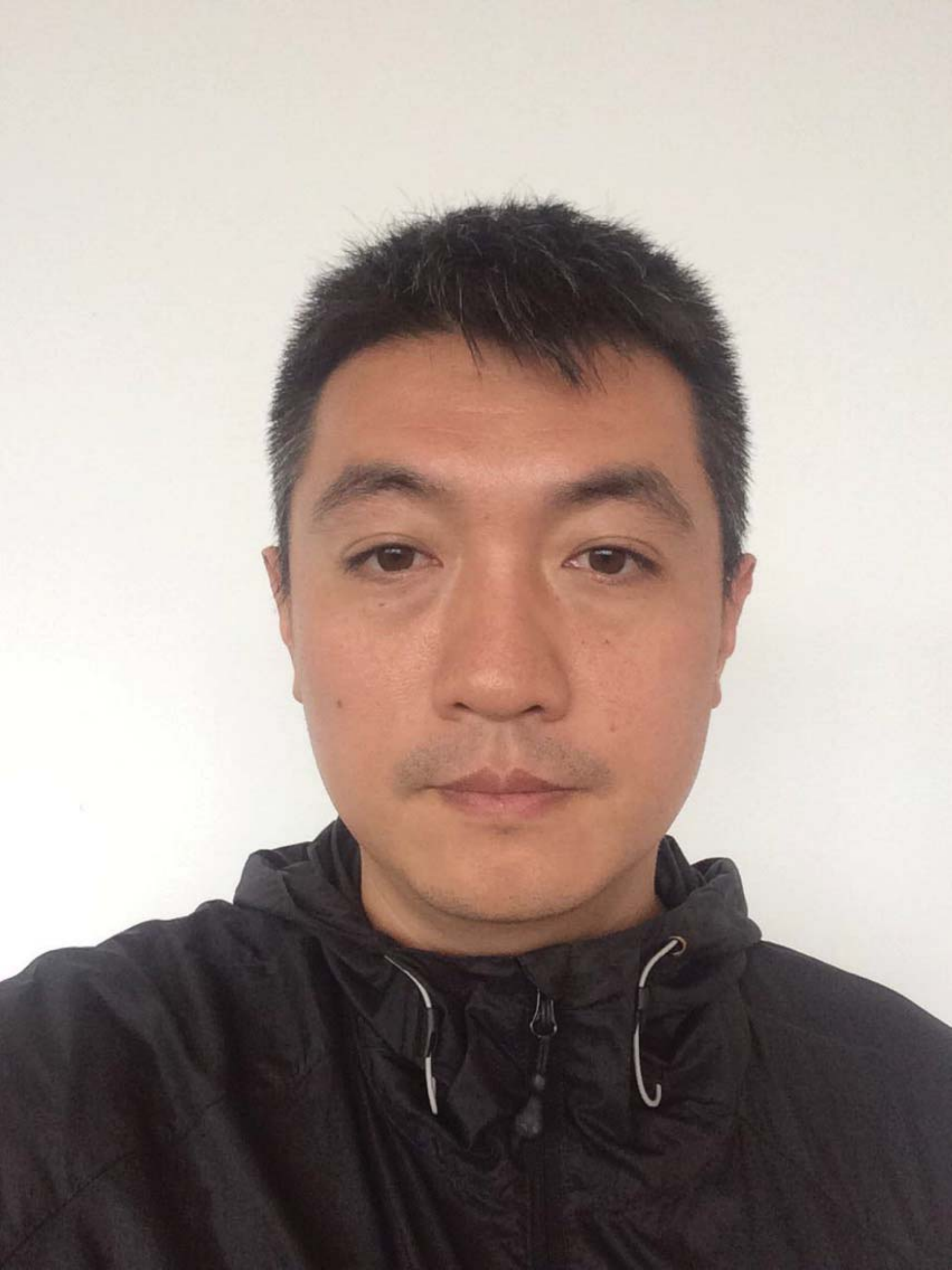}}]{Chao Dong} (Member, IEEE) received his Ph.D degree in Communication Engineering from PLA University of Science and Technology, China, in 2007. He is now a full professor with College of Electronic and Information Engineering, Nanjing University of Aeronautics and Astronautics, China. His current research interests include D2D communications, UAVs swarm networking and anti-jamming network protocol.
\end{IEEEbiography}

\begin{IEEEbiography}[{\includegraphics[width=1in,height=1.25in,clip,keepaspectratio]{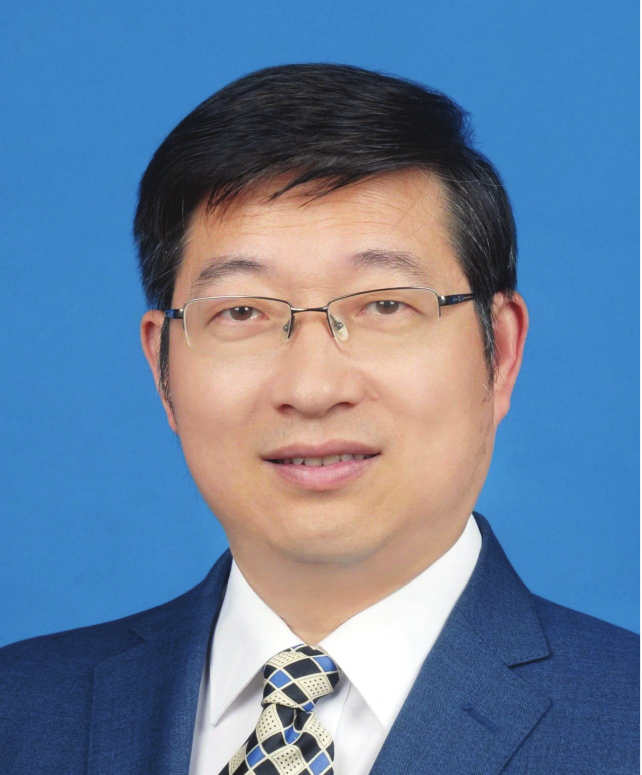}}] {Qihui Wu} (Fellow, IEEE) received the B.S. degree in communications engineering and the M.S. and Ph.D. degrees in communications and information systems from the Institute of Communications Engineering, Nanjing, China, in 1994, 1997, and 2000, respectively. From 2003 to 2005, he was a Post-Doctoral Research Associate with Southeast University, Nanjing. From 2005 to 2007, he was an Associate Professor with the College of Communications Engineering, PLA University of Science and Technology, Nanjing, where he was a Full Professor, from 2008 to 2016. From March 2011 to September 2011, he was an Advanced Visiting Scholar with the Stevens Institute of Technology, Hoboken, NJ, USA. Since May 2016, he has been a Full Professor with the College of Electronic and Information Engineering, Nanjing University of Aeronautics and Astronautics, Nanjing. His current research interests include wireless communications and statistical signal processing, with an emphasis on system design of software defined radio, cognitive radio, and smart radio.
\end{IEEEbiography}

\begin{IEEEbiography}[{\includegraphics[width=1in,height=1.25in,clip,keepaspectratio]{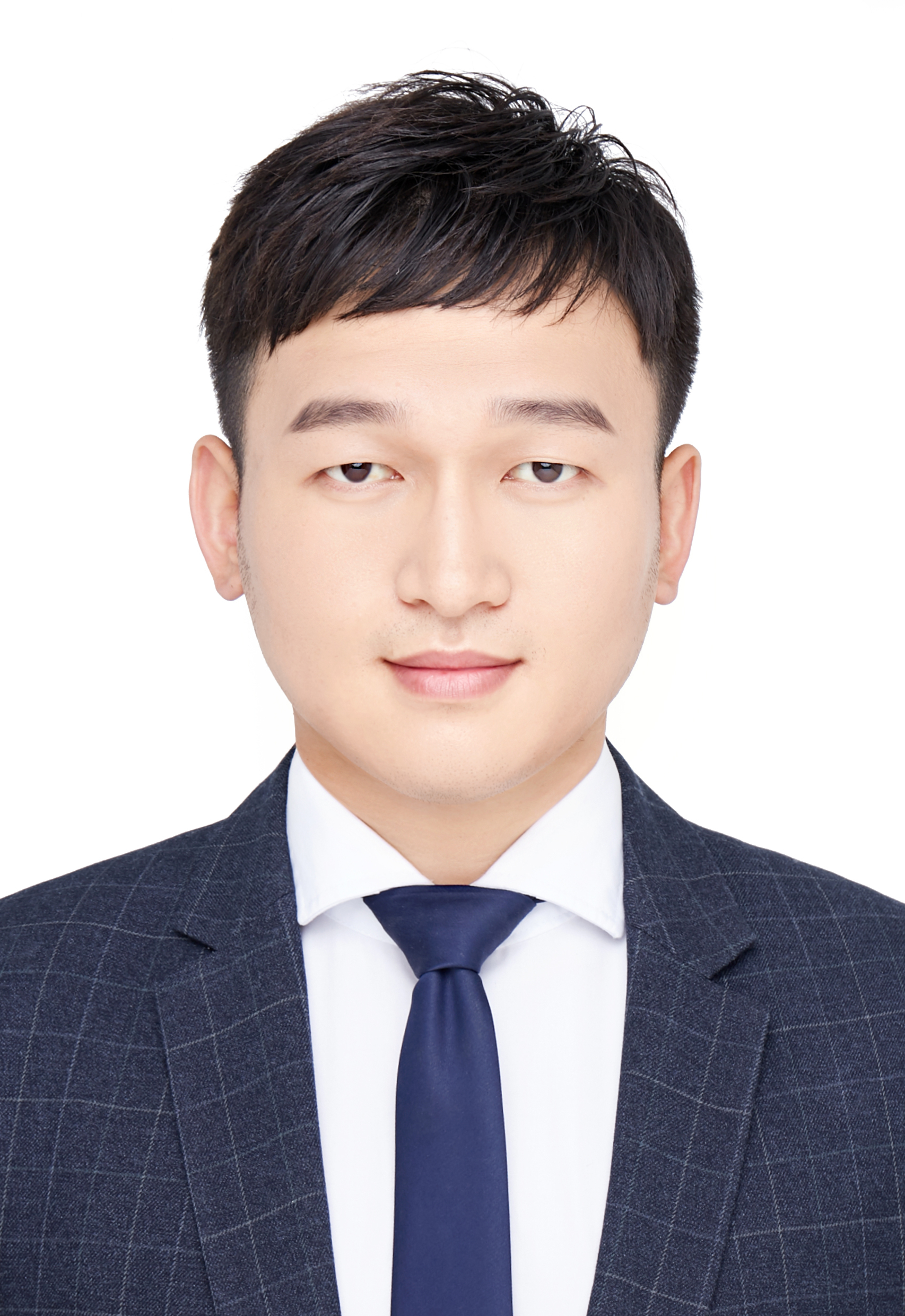}}]{Zhuang Ling} received the B.S. and Ph.D. degrees in the College of Communication Engineering, Jilin University, Jilin, China, in 2016 and 2021, respectively. Currently, he serves as an associate professor in the College of Communication Engineering, Jilin University, Changchun, Jilin, China. He was formerly a postdoctoral fellow in the same college. In 2019, he served as a visiting Ph.D. student in the Department of Electrical and Computer Engineering at the University of Houston. His research interests include Wireless Body Area Network, High-Speed Railway, Backscatter Communications, Energy Harvesting, Age of Information, and Distributionally Robust Optimization.
\end{IEEEbiography}

\begin{IEEEbiography}[{\includegraphics[width=1in,height=1.25in,clip,keepaspectratio]{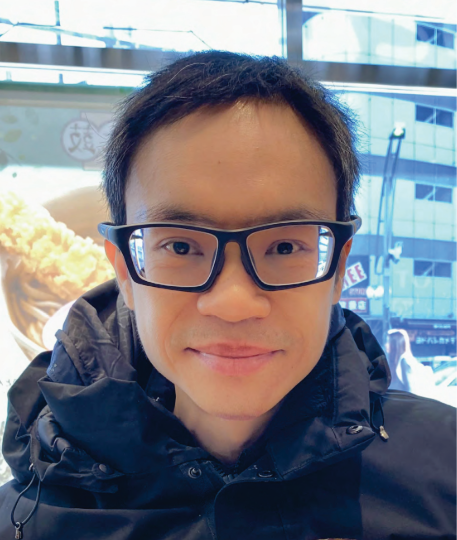}}] {Dusit Niyato} (Fellow, IEEE) is a professor in the College of Computing and Data Science, at Nanyang Technological University, Singapore. He received B.Eng. from King Mongkuts Institute of Technology Ladkrabang (KMITL), Thailand and Ph.D. in Electrical and Computer Engineering from the University of Manitoba, Canada. His research interests are in the areas of mobile generative AI, edge intelligence, quantum computing and networking, and incentive mechanism design.
\end{IEEEbiography}

\begin{IEEEbiography}[{\includegraphics[width=1in,height=1.25in,clip,keepaspectratio]{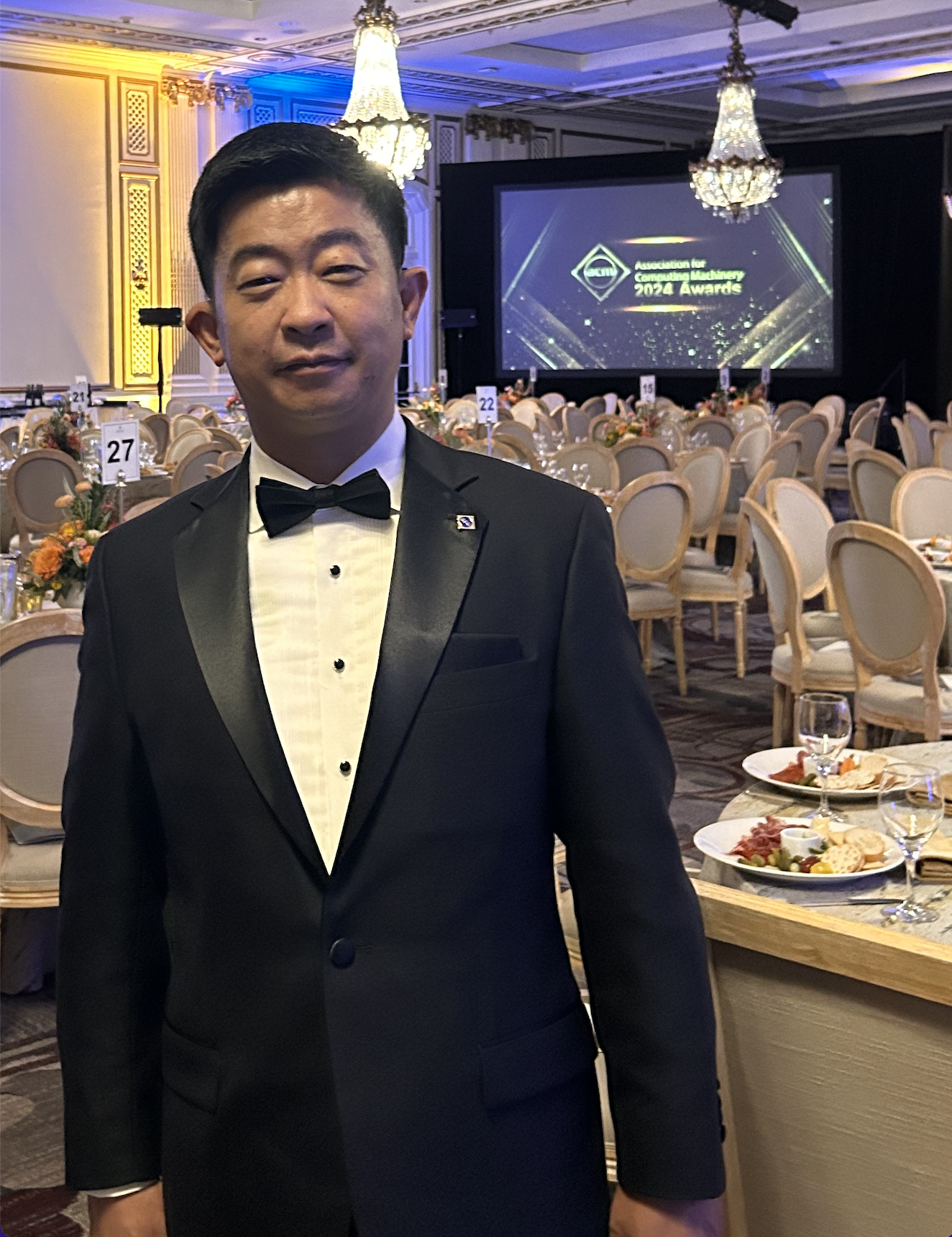}}] {Zhu Han} (Fellow, IEEE) received the B.S. degree in electronic engineering from Tsinghua University, in 1997, and the M.S. and Ph.D. degrees in electrical and computer engineering from the University of Maryland, College Park, in 1999 and 2003, respectively. From 2000 to 2002, he was an R\&D Engineer of JDSU, Germantown, Maryland. From 2003 to 2006, he was a Research Associate at the University of Maryland. From 2006 to 2008, he was an assistant professor at Boise State University, Idaho. Currently, he is a John and Rebecca Moores Professor in the Electrical and Computer Engineering Department as well as in the Computer Science Department at the University of Houston, Texas. Dr. Han's main research targets on the novel game-theory related concepts critical to enabling efficient and distributive use of wireless networks with limited resources. His other research interests include wireless resource allocation and management, wireless communications and networking, quantum computing, data science, smart grid, carbon neutralization, security and privacy.  Dr. Han received an NSF Career Award in 2010, the Fred W. Ellersick Prize of the IEEE Communication Society in 2011, the EURASIP Best Paper Award for the Journal on Advances in Signal Processing in 2015, IEEE Leonard G. Abraham Prize in the field of Communications Systems (best paper award in IEEE JSAC) in 2016, IEEE Vehicular Technology Society 2022 Best Land Transportation Paper Award, and several best paper awards in IEEE conferences. Dr. Han was an IEEE Communications Society Distinguished Lecturer from 2015 to 2018 and ACM Distinguished Speaker from 2022 to 2025, AAAS fellow since 2019, and ACM Fellow since 2024. Dr. Han is a 1\% highly cited researcher since 2017 according to Web of Science. Dr. Han is also the winner of the 2021 IEEE Kiyo Tomiyasu Award (an IEEE Field Award), for outstanding early to mid-career contributions to technologies holding the promise of innovative applications, with the following citation: ``for contributions to game theory and distributed management of autonomous communication networks."
\end{IEEEbiography}

\end{document}


\pagestyle{empty}

\begin{appendices}
	
\section{}\label{appendix}
For a measurable loss function $\phi(\xi):\mathbb{R}^k\rightarrow\mathbb{R}$, the definition of CVaR under safety factor $\alpha$ is
\begin{equation}
	\mathbb{P}-CVaR_\alpha(\phi(\xi))=\underset{\beta\in\mathbb{R}}{inf}\left\{\beta+\frac{1}{1-\alpha}\mathbb{E_P}\left[\phi(\xi)-\beta\right]^+\right\},
\end{equation}
where $\mathbb{P}$ is the probability distribution on $\mathbb{R}^k$. $[\phi(\xi)-\beta]^+$ represents the maximum value between $\phi(\xi)-\beta$ and $0$. $\beta$ is an auxiliary variable on $\mathbb{R}$. Therefore, the worst-case CVaR is
\begin{equation}
	\underset{\mathbb{P}\in\mathcal{P}}{sup}\quad\underset{\beta\in\mathbb{R}}{inf}\left\{\beta+\frac{1}{1-\alpha}\mathbb{E_P}\left[\Theta\xi+\theta^0-\beta\right]^+\right\}.
\end{equation}
According to the Saddle Point Theorem, the operators $\underset{\mathbb{P}\in\mathcal{P}}{sup}$ and $\underset{\beta\in\mathbb{R}}{inf}$ can be interchanged, i.e.,
\begin{equation}
	\begin{split}
		\underset{\mathbb{P}\in\mathcal{P}}{sup}\quad\underset{\beta\in\mathbb{R}}{inf}\left\{\beta+\frac{1}{1-\alpha}\mathbb{E_P}\left[\Theta\xi+\theta^0-\beta\right]^+\right\}=\\
		\underset{\beta\in\mathbb{R}}{inf}\left\{\beta+\frac{1}{1-\alpha}\underset{\mathbb{P}\in\mathcal{P}}{sup}\mathbb{E_P}\left[\Theta\xi+\theta^0-\beta\right]^+\right\}.
	\end{split}
\end{equation}
For any fixed $\Theta$, $\theta^0$ and $\beta$, the optimization problem of inner layer is
\begin{equation}\label{inner-problem}
	\underset{\mathbb{P}\in\mathcal{P}}{sup}\mathbb{E_P}\left[\Theta\xi+\theta^0-\beta\right]^+.
\end{equation}
Let $\varrho =\Theta\xi$, and based on the uncertainty set $\mathcal{P}$, where $\mathcal{P}=\left\{\mathbb{P}\in\mathcal{P}|\mathbb{E_P}(\xi)=\mu_\xi,\mathbb{D_P}(\xi)=\sigma_\xi^{2}\right\}$, the mean value and variance of $\varrho$ are $\Theta\mu_\xi$ and $\Theta^2\sigma_\xi^{2}$, respectively. Hence, the inner layer optimization (\ref{inner-problem}) is equivalent to
\begin{equation}\label{e59}
	\begin{split}
		&\underset{\varphi\in\mathcal{C}}{sup}\quad\int_\mathbb{R}\left(\Theta\xi+\theta^0-\beta\right)^+\varphi\,d\varrho\\
		\textrm{s.t.}\quad&\int_\mathbb{R}\varphi\,d\varrho=1,\\
		&\int_\mathbb{R}\varrho\varphi\,d\varrho=\Theta\mu_\xi,\\
		&\int_\mathbb{R}\varrho^2\varphi\,d\varrho=\Theta^2\mu_\xi^2+\Theta^2\sigma_\xi^2,
	\end{split}
\end{equation}
where $\varphi$ denotes the decision variables and $\mathcal{C}$ represents the cone of nonnegative Borel measures on $\mathbb{R}$. By introducing the dual variables $\chi_1$, $\chi_2$ and $\chi_3$ of the constraints in (\ref{e59}), the Lagrangian dual problem is reformulated as
\begin{align}
	\underset{\chi_1,\chi_2,\chi_3}{inf}&\quad \chi_1+\chi_2\Theta\mu_\xi+\chi_3\left(\Theta^2\mu_\xi^2+\Theta^2\sigma_\xi^2\right)\nonumber\\
	\textrm{s.t.}\quad&\chi_1+\chi_2 \varrho+\chi_3\varrho^2\geq0,\label{e60}\\
	&\chi_1+\chi_2\varrho+\chi_3\varrho^2\geq\Theta+\varrho-\beta,\label{e61}\\
	&\chi_3>0.
\end{align}
Then, $\varrho^*=-\frac{\chi_2}{2\chi_3}$ is substituted into (\ref{e60}) and (\ref{e61}), and based on the Strong Duality Theorem, (\ref{e60}) and (\ref{e61}) can be equivalently transformed into
\begin{equation}
	\chi_1-\frac{\chi_2^2}{4\chi_3}\geq 0,
\end{equation}
and
\begin{equation}
	\chi_1+\beta-\Theta-\frac{\left(\chi_2-1\right)^2}{4\chi_3}\geq 0.
\end{equation}
Considering $\chi_1=e+\frac{1}{4z}\left(q-\Theta\mu_\xi\right)^2$, $\chi_2=\frac{q-\Theta\mu_\xi}{2z}$ and $\chi_3=\frac{1}{4z}$, (\ref{e59}) is reformulated as
\begin{equation}
	\begin{split}
		\underset{\beta,e,q,z,s}{inf}& e+s ,\\
		&e-\theta ^0+\beta +q-\Theta \mu_\xi -z>0,\\
		&e\geq 0,z> 0,\\
		&4zs\geq q^2+\Theta^2\sigma_\xi^2.
	\end{split}
\end{equation}
Consequently, it is proved that the worst-case CVaR under the safety factor $\alpha$ can be approximated by a tractable convex programming problem as
\begin{equation}
	\begin{split}
		\underset{\beta,e,q,z,s}{inf}&\beta +\frac{1}{1-\alpha}\left( e+s \right) ,\\
		&e-\theta ^0+\beta +q-\Theta \mu_\xi -z>0,\\
		&e\geq 0,z> 0,\\
		&\begin{Vmatrix}
				q\\
				\Theta\sigma_\xi\\
				z-s
		\end{Vmatrix} \leq z+s.\\
	\end{split}
\end{equation}
\end{appendices}